\documentclass[journal,12pt,onecolumn,draftclsnofoot]{IEEEtran}
\usepackage[latin9]{inputenc}
\usepackage{amsmath}
\usepackage{amssymb}
\usepackage{graphicx}
\usepackage{lipsum}
\usepackage{savesym}
\usepackage{slashbox}
\usepackage{amsthm}
\usepackage{cuted}
\usepackage{nccmath}\usepackage{amsfonts}
\usepackage{epstopdf}
\usepackage{lipsum}
\usepackage{float}
\usepackage{stfloats}
\usepackage{color}

\makeatletter

\providecommand{\tabularnewline}{\\}

\providecommand{\tabularnewline}{\\}

\usepackage{lipsum}
\usepackage{savesym}
\usepackage{slashbox}
\usepackage{amsthm}
\usepackage{cuted}
\usepackage{nccmath}\usepackage{amsfonts}
\usepackage{lipsum}
\usepackage{float}
\usepackage{stfloats}
\usepackage{anyfontsize}
\usepackage{color}
\usepackage{amsmath,amssymb}
\usepackage{amsfonts}

\providecommand{\U}[1]{\protect\rule{.1in}{.1in}}
\floatstyle{plain}

\newfloat{twocolequfloat}{b}{zzz}
\floatname{twocolequfloat}{Equation}

\newtheorem{proposition}{Proposition}\newtheorem{remark}{Remark}\iffalse
\makeatother

\begin{document}
	\title{New Accurate Approximation for {Average Error Probability} Under
		$\kappa-\mu$ Shadowed Fading Channel}
	
	\author{{Yassine Mouchtak, Faissal El\ Bouanani\\
			ENSIAS, Mohammed V University in Rabat\\
			e-mails: \{yassine.mouchtak,f.elbouanani\}@um5s.net.ma\\
			Corresponding author: Faissal El Bouanani.}}
	
	\maketitle
	\author{{Yassine Mouchtak, Faissal El\ Bouanani, \textit{Senior Member,
				IEEE} \thanks{ Y. Mouchtak and F. El\ Bouanani are with Mohammed V University,
				10000 Rabat, Morocco (e-mails: {faissal.elbouanani@um5.ac.ma, yassine.mouchtak}@um5s.net.ma).}}}
	
	\begin{abstract}
		This paper proposes new accurate approximations for average
		error probability {(AEP)} of a communication system employing either
		$M$-phase-shift keying (PSK) or differential quaternary PSK with Gray coding(GC-DQPSK)
		modulation schemes over $\kappa-\mu$ shadowed fading
		channel. Firstly, new accurate approximations of error probability (EP)
		of both modulation schemes are derived over additive white Gaussian
		noise (AWGN) channel. Leveraging the trapezoidal integral method,
		a tight approximate expression of {symbol error probability} for $M$-PSK modulation is presented, while
		new upper and lower bounds for Marcum $Q$-function of the first order
		(MQF), and subsequently those for {bit error probability (BER)} under DQPSK scheme, are proposed.
		Next, these bounds are linearly combined to propose a highly refined and accurate
		{BER's} approximation. The key idea manifested in the decrease property
		of modified Bessel function $I_{v}$, strongly related to MQF, with
		its argument $v$. Finally, theses approximations are used to tackle
		{AEP's} approximation under $\kappa-\mu$ shadowed fading. Numerical
		results show the accuracy of the presented approximations compared
		to the exact ones.
	\end{abstract}
	

	\section{Introduction}
	
	Wireless technologies are becoming part of our daily lives and their
	utilization increase rapidly due to many advantages such as cost-effectiveness,
	global coverage and flexibility. Nevertheless, these technologies
	are infected by many phenomena including shadowing which is relatively
	slow and gives rise to long-term signal variations and multipath fading
	which is due to constructive and destructive interferences as a result
	of delayed, diffracted, reflected, and scattered signal components
	\cite{gqfunc}. A great number of communication channels' models have
	been proposed in the literature to describe either the fading or the
	joint shadowing/fading phenomena \cite{m1}-\cite{m4}. Recently,
	the $\kappa-\mu$ shadowed fading proposed in \cite{kappa}, has attracted
	a lot of interest due to its versatility and wide applicability in
	practical scenarios. For instance, it was used for characterizing signal
	reception in device-to-device communications, body-to-body communications,
	underwater acoustic, fifth-generation (5G) communications, and satellite
	communication systems \cite{Ap1}-\cite{ap5}. In addition, it was
	shown that numerous statistical models can be derived from the $\kappa-\mu$
	shadowed one by setting the parameters to some specific real positive
	values \cite{kapamu}. {Particularly, when the parameters $\mu$ and $m$ are integer, such a model is equivalent to what's referred to as composite fading, namely, mixture Gamma distribution \cite{kapa2}.}
	
	The {average error probability (AEP)} is a fundamental performance evaluation
	tool in digital communications, quantifying the reliability of an
	instantaneous received signal. Furthermore, dealing with the average
	{EP (AEP)} is quite practical in most applications as it states the
	average performance irrespective of time. Nonetheless, evaluating
	{AEP} in closed form remains a big challenge for numerous communication
	systems because of the complexity of either the end-to-end fading
	model or the employed modulation technique.  Essentially, depending
	on the employed modulation scheme, {EP} is provided in either complicated
	integral form \cite{gqfunc} or first-order Marcum $Q$-function (MQF)
	and the zeroth-order modified Bessel function (MBF) of the first kind
	\cite{proakis} for various $M$-ary and differential quadrature phase-shift
	keying (DQPSK) modulation schemes, respectively. That integral form
	can be reexpressed also in terms Gaussian $Q$-function (GQF), which
	is not known in closed form. By its turn, the MQF integral-form involves
	the MBF with exponential term \cite{gqfunc}, that can be rewritten
	appropriately as an upper incomplete upper Fox's H-function (UIFH),
	or equivalently, an infinite summation of the product of upper incomplete
	Gamma functions \cite{table integral}. Thus, obtaining {AEP} requires
	the averaging of a UIFH over a generalized fading distribution, which
	is not evident particularly for fading model with probability density
	function (PDF) involving the product of exponential and Fox's H- functions
	(e.g. $\kappa-\mu$ shadowed model). Obviously, deriving accurate
	bounds or approximations for the {AEP} is strongly depending on the
	{EP's} ones. To this end, several EP's bounds and approximations for
	the EP are proposed in the literature, for instance, in \cite{chie}-\cite{luca}, numerous bounds for the symbol error probability (SEP) in the case of $M$-ary PSK modulation are derived in terms of GQF and its powers. Such a function is itself mathematically intractable when involved in complicated integrals resulting from generalized fading distributions. To remedy this problem, several various works deal with simple, and accurate approximations or bounds for GQF when applied to inspect the performance of a communication system experiencing to a particular bivariate-Fox's H-fading model  \cite{IeeeACCESS}, \cite{yassine}. In contrast, evaluating the performance of GC-DQPSK modulation requires simple bounds or approximate expressions for MQF due to its complicated closed-form and intractability when involved in the computation of AEP \cite{12}-\cite{15}. In \cite{ferrari} and \cite{sun}, bounds
	for {EP} are investigated, while in \cite{barzic}, new lower and upper
	bounds for EP were proposed, based on which a novel approximation
	was derived. Despite the good accuracy of the latter's bounds and/or approximation for both MPSK and GC-DQPSK,
	they remain useless for {AEP} computation because of their forms' complications.
	
	\subsection{Motivation}
	
	The performance of wireless communication systems, with perfect channel
	state information (CSI) knowledge at the receiver, is widely examined
	by the scientific community. However, imperfect estimation of channel
	coefficients is dealt with various practical scenarios, leading to a
	significant degradation of the system performance. To overcome this
	limitation, differential modulation (DM) can be considered as an alternative
	solution particularly for low-power wireless systems, such as wireless
	sensor networks and relay networks \cite{abouie}. The main advantage
	of this scheme is its simplicity of detection due to the unnecessary
	channel coefficients estimation and tracking, leading to significant
	reduction in the receiver computational complexity \cite{24}, \cite{25}.
	However, this comes at a cost of higher error rate or lower spectral
	efficiency. As a result, selecting the most suitable modulation scheme
	depends on the considered application and both coherent and non-coherent
	detections. To this end, this paper is devoted to analyzing the performance
	of two modulation schemes, namely $M$-PSK and DQPSK over $\kappa-\mu$
	shadowed fading channel.
	
	\subsection{Contribution}
	
	Capitalizing on the above, we aim at this work to propose accurate
	approximations for {AEP} under $\kappa-\mu$ shadowed fading and aforementioned
	modulation schemes. Specifically, utilizing the trapezoidal integral
	method, the {EP} integral form for various $M$-ary modulation schemes
	is tightly approximated particularly for $M$-PSK scheme, while for
	DQPSK technique, we start by deriving simple lower and upper bounds
	for {EP} by bounding MQF, to be used jointly in finding {AEP} for generalized
	fading models.
	
	Pointedly, our main key contributions can be summarized as follows
	\begin{itemize}
		\item We propose a new exponential type approximation for the {EP's} first
		form applied to$M$-PSK modulation by using the trapezoidal technique
		integral. To the best of the authors' knowledge, such accurate {EP's}
		approximation outperforms those presented in the literature,
		\item We derive new upper and lower bounds of {EP} in the case of DQPSK modulation
		based on which an accurate approximation of SEP is proposed,
		\item We provide, relying on the two proposed {EP's} approximations, a tight
		approximate expression for {AEP} over $\kappa-\mu$ shadowed fading
		channel,
		\item We provide the asymptotic analysis for both forms of {AEP} and we demonstrate
		that the diversity order over $\kappa-\mu$ shadowed fading channel
		remains constant.
	\end{itemize}
	Motivated by this introduction, the rest of this paper can be structured
	as follows. Methods and analysis used are described briefly in Section II. In section III, a new approximation for the first {EP} form
	(i.e., $M$-ary modulation) is presented for $M$-PSK while, new lower
	and upper bound of {EP} in the case of DQPSK are derived, based on
	which an accurate approximation for the {EP} is deduced. In Section
	IV, the expression of {AEP} under $\kappa-\mu$ shadowed fading for
	both modulation schemes is evaluated. In section V, the respective
	results are illustrated and verified by comparison with the exact
	ones using simulation computing. Section VI summarizes the main conclusions.
	\section{{Methods}}
	The present work deal with the performance analysis of a wireless communication system subject to $\kappa-\mu$ shadowed fading. To this end, various methods are applied to derive new accurate approximate expressions for the error probability. Pointedly, the Trapezoedial integral technique is utilized to propose a simple approximation for the SEP in the case of $M$-PSK modulation, while bounding technique is used to derive upper and lower bounds for the bit error probability (BEP) in the case of GC-DQPSK modulation. Furthermore, by using Matlab curve fitting application, a refinement is applied on these two new bounds to obtain a novel accurate approximation for the BEP of this latter modulation scheme. Next, the approximations of the EP are used to derive tight approximation for the AEP under the considered fading model. Lastly, the accuracy of the derived AEP's approximation is validated using Monte carlo simulation.

	\section{Bounds on the SEP}
	
	In this section, we propose new approximate expressions for the two
	potential different forms of {EP}, namely (i) complicated integral
	form, and (ii) MQF form, applied to $M$-PSK and DQPSK modulations
	with Gray coding, respectively.
	
	\subsection{EP with integral form}
	
	\begin{proposition} {The SEP} for $M$-PSK modulation can be tightly
		approximated by
		\begin{equation}
		\widetilde{\mathcal{H}}_{1}\left(\gamma\right)\simeq\sum_{l=1}^{7}\mathcal{A}_{l}\exp\left(-\mathcal{B}_{l}\gamma\right),\label{APProx-MPSK}
		\end{equation}
		while $\mathcal{A}_{l}$ and $\mathcal{B}_{l}$ \ are given in Table
		\ref{coeffMPSK}.
		
		\begin{table*}[tbh]
			\caption{The coefficients $\mathcal{A}_{l}$ and $\mathcal{B}_{l}$.}
			\label{coeffMPSK} \centering{}%
			\begin{tabular}{c|c|c|c|c|c|c|c}
				\hline
				$l$  & $1$  & $2$  & $3$  & $4$  & $5$  & $6$  & $7$\tabularnewline
				\hline
				$\mathcal{A}_{l}$  & $\frac{7M-8}{48M}$  & $\frac{1}{8}$  & $\frac{1}{8}$  & $\frac{1}{8}$  & $\frac{M-2}{12M}$  & $\frac{M-2}{6M}$  & $\frac{M-2}{6M}$\tabularnewline
				\hline
				$\mathcal{B}_{l}$  & $\varrho$  & $2\varrho$  & $\frac{20\varrho}{3}$  & $\frac{20\varrho}{17}$  & $\frac{\varrho}{\cos^{2}(\frac{M-2}{2M}\pi)}$  & $\frac{\varrho}{\cos^{2}(\frac{M-2}{6M}\pi)}$  & $\frac{\varrho}{\cos^{2}(\frac{M-2}{3M}\pi)}$\tabularnewline
				\hline
			\end{tabular}
		\end{table*}
		
	\end{proposition}
	
	\begin{proof}
		
		The SEP for $M$-PSK modulation is given as \cite[Eq. (8.22)]{gqfunc}
		
		\begin{equation}
		\mathcal{H}_{1}\left(\gamma\right)=\frac{1}{\pi}\int_{0}^{\frac{M-1}{M}\pi}\exp\left(-\frac{\varrho\gamma}{\sin^{2}(\theta)}\right)d\theta,\label{CBERMPSK}
		\end{equation}
		with
		\begin{equation}
		\varrho=\log_{2}(M)\sin^{2}\left(\frac{\pi}{M}\right),
		\end{equation}
		and $\gamma$ denotes the signal-to-noise (SNR) ratio per bit.
		
		Subsequently, $\left(\ref{CBERMPSK}\right)$ can be written as
		
		\begin{equation}
		\mathcal{H}_{1}\left(\gamma\right)=Q\left(\sqrt{2\varrho\gamma}\right)+\underbrace{\frac{1}{\pi}\int_{0}^{\frac{M-2}{2M}\pi}\exp\left(-\frac{\varrho\gamma}{\cos^{2}(t)}\right)dt}_{\mathcal{I}}\label{CMPSK2}
		\end{equation}
		where $Q\left(.\right)$ denotes the Gaussian $Q-$Function \cite[Eq. (4.1)]{gqfunc}.
		
		The integral $\mathcal{I}$ can be approximated using numerical integration
		rules. The trapezoidal rule for definite integration of an arbitrary function
		between $[x_{0},x_{0}+n\phi]$ is given by
		\[
		\int_{x_{0}}^{x_{0}+n\phi}f\left(t\right)dt=\frac{\phi}{2}\left[g_{0}+g_{n}+2\sum_{i=1}^{n-1}g_{i}\right],
		\]
		where $g_{i}=f\left(x_{0}+i\phi\right)$ for $i=0..n$, $n$ refers
		to the number of sub-intervals equally spaced trapeziums, and $\phi$
		defines the spacing. Note that greater $\ n$ the higher the accuracy's
		approximation and the computational complexity as well.
		
		By setting $n=3$ and $f\left(t\right)=\exp\left[-\varrho\gamma/\cos^{2}(t)\right]$,
		$\mathcal{I}$ can be approximated as
		\begin{eqnarray}
		\mathcal{I} & \simeq & \frac{M-2}{12M}\left[\begin{array}{c}
		\exp\left(-\varrho\gamma\right)\\
		+\exp\left(-\frac{\varrho\gamma}{\cos^{2}(\frac{M-2}{2M}\pi)}\right)
		\end{array}\right]\nonumber \\
		&  & +\frac{M-2}{6M}\left[\begin{array}{c}
		\exp\left(-\frac{\varrho\gamma}{\cos^{2}(\frac{M-2}{6M}\pi)}\right)\\
		+\exp\left(-\frac{\varrho\gamma}{\cos^{2}(\frac{M-2}{3M}\pi)}\right)
		\end{array}\right].\label{inegral1}
		\end{eqnarray}
		
		By plugging \cite[Eq. (8b)]{sadhwani} and $\eqref{inegral1}$ in $\eqref{CMPSK2},$
		\eqref{besseli0.5} is attained.
		
	\end{proof}
	
	Table \ref{compaMPSK} confirms the accuracy of the proposed approximation.
	Interestingly, one can ascertain that this tightness can be further
	improved by increasing $n$, i.e., by increasing the number of terms
	in the approximation.
	
	\begin{table*}
		\caption{Comparison between the exact and approximate {SEP} for $M$-PSK modulation.}
		\label{compaMPSK}%
		\begin{tabular}{c|c|c|c|c|c|c|c}
			\multicolumn{8}{c}{$M=4$}\tabularnewline
			\hline
			\hline
			\backslashbox{SEP}{$\gamma$}  & $1$  & $2$  & $3$  & $4$  & $5$  & $6$  & $7$ \tabularnewline
			\hline
			$\mathcal{H}_{1}\left(\gamma\right),$ Eq. $\left(\ref{CBERMPSK}\right)$  & $1.508\times10^{-1}$  & $4.494\times10^{-2}$  & $1.425\times10^{-2}$  & $4.672\times10^{-3}$  & $1.565\times10^{-3}$  & $5.319\times10^{-4}$  & $1.828\times10^{-4}$ \tabularnewline
			\hline
			$\widetilde{\mathcal{H}}_{1}\left(\gamma\right),$ Eq. $\left(\ref{APProx-MPSK}\right)$  & $1.501\times10^{-1}$  & $4.460\times10^{-2}$  & $1.414\times10^{-2}$  & $4.642\times10^{-3}$  & $1.556\times10^{-3}$  & $5.289\times10^{-4}$  & $1.816\times10^{-4}$ \tabularnewline
			\hline
			\multicolumn{8}{c}{$M=8$}\tabularnewline
			\hline
			\hline
			\backslashbox{SEP}{$\gamma$}  & \multicolumn{1}{c|}{$2$} & $4$  & $6$  & $8$  & $10$  & $14$  & $16$ \tabularnewline
			\hline
			$\mathcal{H}_{1}\left(\gamma\right),$ Eq. $\left(\ref{CBERMPSK}\right)$  & \multicolumn{1}{c|}{$1.849\times10^{-1}$} & $6.083\times10^{-2}$  & $2.167\times10^{-2}$  & $8.018\times10^{-3}$  & $3.034\times10^{-3}$  & $4.526\times10^{-4}$  & $1.772\times10^{-4}$ \tabularnewline
			\hline
			$\widetilde{\mathcal{H}}_{1}\left(\gamma\right),$ Eq. $\left(\ref{APProx-MPSK}\right)$  & \multicolumn{1}{c|}{$1.847\times10^{-1}$} & $6.077\times10^{-2}$  & $2.156\times10^{-2}$  & $7.976\times10^{-3}$  & $3.022\times10^{-3}$  & $4.504\times10^{-4}$  & $1.760\times10^{-4}$ \tabularnewline
			\hline
			\multicolumn{8}{c}{$M=16$}\tabularnewline
			\hline
			\hline
			\backslashbox{SEP}{$\gamma$}  & \multicolumn{1}{c|}{$5$} & $10$  & $20$  & $30$  & $35$  & $40$  & $45$ \tabularnewline
			\hline
			$\mathcal{H}_{1}\left(\gamma\right),$ Eq. $\left(\ref{CBERMPSK}\right)$  & \multicolumn{1}{c|}{$2.173\times10^{-1}$} & $8.100\times10^{-2}$  & $1.360\times10^{-2}$  & $2.508\times10^{-3}$  & $1.096\times10^{-3}$  & $4.832\times10^{-4}$  & $2.143\times10^{-4}$ \tabularnewline
			\hline
			$\widetilde{\mathcal{H}}_{1}\left(\gamma\right),$ Eq. $\left(\ref{APProx-MPSK}\right)$  & \multicolumn{1}{c|}{$2.177\times10^{-1}$} & $8.072\times10^{-2}$  & $1.358\times10^{-2}$  & $2.500\times10^{-3}$  & $1.091\times10^{-3}$  & $4.794\times10^{-4}$  & $2.122\times10^{-4}$ \tabularnewline
			\hline
		\end{tabular}
	\end{table*}

	\subsection{EP with MQF form}
	
	{The bit error probability (BEP)} for DQPSK modulation with Gray coding is given by \cite{proakis}
	\begin{equation}
	\mathcal{H}_{2}\left(\gamma\right)=Q_{1}\left(a\sqrt{\gamma},b\sqrt{\gamma}\right)-\frac{1}{2}I_{0}\left(\sqrt{2}\gamma\right)\exp\left(-2\gamma\right),\label{pep}
	\end{equation}
	with $\gamma$ denotes the SNR per bit, $a=\sqrt{2\left(1-\sqrt{0.5}\right)}$,
	$b=\sqrt{2\left(1+\sqrt{0.5}\right)}$, $I_{v}\left(.\right)$ is
	the $v$-th order modified Bessel function of the first kind \cite[ Eq. (8.431)]{table integral},
	and $Q_{1}\left(.,.\right)$ represents the first-order MQF defined
	as \cite[Eq. (4.34)]{gqfunc}
	
	\begin{equation}
	Q_{1}\left(\alpha,\beta\right)=\int_{\beta}^{\infty}t\exp\left(-\frac{t^{2}+\alpha^{2}}{2}\right)I_{0}\left(\alpha t\right)dt.\label{mqf}
	\end{equation}

	\subsubsection{New lower bound for {BEP}}
	
	\begin{proposition} The {BEP} for DQPSK modulation with Gray coding
		can be lower bounded as
		\begin{equation}
		\mathcal{H}_{2}\left(\gamma\right)\geq\mathcal{L}\left(\gamma\right),
		\end{equation}
		with
		\begin{equation}
		\mathcal{L}\left(\gamma\right)\triangleq\delta\mathcal{K}(a,b,\gamma)-\frac{1}{2}I_{0}\left(\sqrt{2}\gamma\right)\exp\left(-2\gamma\right),\label{lower}
		\end{equation}
		\begin{equation}
		\mathcal{K}(a,b,\gamma)=Q\left(\left(b-a\right)\sqrt{\gamma}\right)-Q\left(\left(b+a\right)\sqrt{\gamma}\right),\label{KK}
		\end{equation}
		and $Q(.)$ denotes the Gaussian $Q$-function \cite[Eq. (4.1)]{gqfunc},
		and $\delta=\sqrt{\frac{b}{a}}.$ \end{proposition}
	
	\begin{proof} As $I_{v}$ is a decreasing function with respect to
		the index $v$ \cite{monotonbessel}, yields
		\begin{equation}
		Q_{1}\left(a\sqrt{\gamma},b\sqrt{\gamma}\right)\geqslant\mathcal{J},\label{Jiii}
		\end{equation}
		with
		\begin{equation}
		\mathcal{J}\triangleq\int_{b\sqrt{\gamma}}^{\infty}t\exp\left(-\frac{t^{2}+a^{2}\gamma}{2}\right)I_{\frac{1}{2}}\left(a\sqrt{\gamma}t\right)dt.\label{J-integral}
		\end{equation}
		
		Now applying \cite[Eq. (8.431.4)]{table integral} for $v=\frac{1}{2}$,
		one can ascertain
		\begin{equation}
		I_{\frac{1}{2}}\left(a\sqrt{\gamma}t\right)=\frac{2\sinh(a\sqrt{\gamma}t)}{\sqrt{2\pi a\sqrt{\gamma}t}},\label{besseli0.5}
		\end{equation}
		where $\sinh$$\left(.\right)$ accounts for the hyperbolic sine function.
		
		By plugging $\left(\ref{besseli0.5}\right)$ into $\left(\ref{J-integral}\right)$
		along with the following identity
		
		\begin{equation}
		\sinh\left(a\sqrt{\gamma}t\right)=\frac{\exp\left(a\sqrt{\gamma}t\right)-\exp\left(-a\sqrt{\gamma}t\right)}{2},\label{sinh}
		\end{equation}
		one can obtain
		\begin{equation}
		\mathcal{J}\geq\frac{1}{\sqrt{2\pi a}}\int_{b\sqrt{\gamma}}^{\infty}\sqrt{\frac{t}{\sqrt{\gamma}}}\left[\begin{array}{c}
		\exp\left(-\frac{\left(t-a\sqrt{\gamma}\right)^{2}}{2}\right)\\
		-\exp\left(-\frac{\left(t+a\sqrt{\gamma}\right)^{2}}{2}\right)
		\end{array}\right]dt.
		\end{equation}
		
		Finally, as $t\geq b\sqrt{\gamma}$, yields
		
		\begin{eqnarray}
		\mathcal{J} & \geq & \frac{\delta}{\sqrt{2\pi}}\left[\begin{array}{c}
		\int_{b\sqrt{\gamma}}^{\infty}\exp\left(-\frac{\left(t-a\sqrt{\gamma}\right)^{2}}{2}\right)dt\\
		-\int_{b\sqrt{\gamma}}^{\infty}\exp\left(-\frac{\left(t+a\sqrt{\gamma}\right)^{2}}{2}\right)dt
		\end{array}\right]\\
		& = & \frac{\delta}{\sqrt{2\pi}}\mathcal{K}(a,b,\gamma),\nonumber
		\end{eqnarray}
		which concludes the proof. \end{proof}
	
	\subsubsection{New upper bound for {BEP}}
	
	\begin{proposition} For $\gamma\geq0,$ holds
		\[
		\mathcal{H}_{2}\left(\gamma\right)\leq\mathcal{U}\left(\gamma\right),
		\]
		with
		\begin{equation}
		\mathcal{U}\left(\gamma\right)\triangleq\frac{1}{\delta}\mathcal{K}(a,b,\gamma)+\frac{1}{2}I_{0}\left(\sqrt{2}\gamma\right)\exp\left(-2\gamma\right).\label{upper}
		\end{equation}
	\end{proposition}
	
	\begin{proof} Relying on $\left(\ref{mqf}\right)$ and using integration
		by part by considering $u^{\prime}(t)=t\exp\left(-\frac{t^{2}}{2}\right)$
		and $w=I_{0}\left(\sqrt{a}\gamma t\right)$, one can see
		\[
		Q_{1}\left(a\sqrt{\gamma},b\sqrt{\gamma}\right)=\mathcal{T}+I_{0}\left(\sqrt{2}\gamma\right)\exp\left(-2\gamma\right),
		\]
		
		with
		
		\begin{equation}
		\mathcal{T=}\int_{b\sqrt{\gamma}}^{\infty}a\sqrt{\gamma}\exp\left(-\frac{t^{2}+a^{2}\gamma}{2}\right)I_{1}\left(a\sqrt{\gamma}t\right)dt\cdot\label{aa}
		\end{equation}
		
		Again, by incorporating the inequality $I_{1}\left(t\right)\leq I_{\frac{1}{2}}\left(t\right)$
		\cite{monotonbessel},\ alongside with $\left(\ref{besseli0.5}\right)$
		and $\left(\ref{sinh}\right)$ into $\left(\ref{aa}\right)$, we get
		\begin{equation}
		\mathcal{T\leq}\sqrt{\frac{a\sqrt{\gamma}}{2\pi}}\int_{b\sqrt{\gamma}}^{\infty}\frac{1}{\sqrt{t}}\left[\begin{array}{c}
		\exp\left(-\frac{\left(t-a\sqrt{\gamma}\right)^{2}}{2}\right)\\
		-\exp\left(-\frac{\left(t+a\sqrt{\gamma}\right)^{2}}{2}\right)
		\end{array}\right]dt.
		\end{equation}
		
		Moreover, as $t\geq b\sqrt{\gamma}$ in the aforementioned integrand,
		one can check
		\begin{equation}
		\mathcal{T}\leq\frac{1}{\delta}\mathcal{K}(a,b,\gamma).\label{aa3}
		\end{equation}
		
		Therefore $\left(\ref{upper}\right)$ can be inferred from $\left(\ref{aa3}\right)$
		jointly with $\left(\ref{KK}\right)$; this completes the proof. \end{proof}
	
	\subsubsection{Approximate {BEP} for DQPSK}
	
	In this part, a tight approximate expression for the {BEP} under DQPSK
	scheme is derived based on the two bounds presented above. In a similar
	manner to the approach followed in \cite{barzic}, the new proposed
	approximation is a linear combination of the two aforementioned bounds
	for $\mathcal{H}_{2}\left(\gamma\right)$, namely
	\begin{equation}
	\widetilde{\mathcal{H}}_{2}\left(\gamma\right)=\widetilde{\chi}\left(\gamma\right)\mathcal{U}\left(\gamma\right)+\left(1-\widetilde{\chi}\left(\gamma\right)\right)\mathcal{L}\left(\gamma\right).\text{ }\label{approx}
	\end{equation}
	
	\begin{proposition} The function $\chi\left(\gamma\right)$ can be
		chosen as
		\begin{equation}
		\widetilde{\chi}\left(\gamma\right)=\mathcal{C}_{0}\exp\left(-\mathcal{D}_{0}\gamma\right)+\mathcal{C}_{1}\exp\left(-\mathcal{D}_{1}\gamma\right),\label{formof X}
		\end{equation}
		where $\mathcal{C}_{i}$ and $\mathcal{D}_{i}$ are the best-fit parameters,
		depending on the SNR interval, summarized in Table \ref{optimumparametr}.
	\end{proposition}
	
	\begin{table}[tbh]
		\caption{Optimum values of fitting parameters for different SNR ranges.}
		\label{optimumparametr} \centering{}%
		\begin{tabular}{c|c|c|c|c}
			\hline
			\backslashbox{SNR range}{$\mathcal{C}_{i}$, $\mathcal{D}_{i}$}  & $\mathcal{C}_{0}$  & $\mathcal{D}_{0}$  & $\mathcal{C}_{1}$  & $\mathcal{D}_{1}$\tabularnewline
			\hline
			$\gamma<1$  & $0.1786$  & $2.903$  & $0.7564$  & $0.1307$ \tabularnewline
			\hline
			$1\leq\gamma<8$  & $0.3798$  & $1.895$  & $0.6183$  & $7.93\times10^{-4}$ \tabularnewline
			\hline
			$\gamma>8$  & $0.005206$  & $0.2764$  & $0.6146$  & $5.593\times10^{-5}$ \tabularnewline
			\hline
		\end{tabular}
	\end{table}
	
	\begin{proof} First, note that the following function
		\begin{equation}
		\chi\left(\gamma\right)=\frac{\mathcal{H}\left(\gamma\right)-\mathcal{L}\left(\gamma\right)}{\mathcal{U}\left(\gamma\right)-\mathcal{L}\left(\gamma\right)},\label{XX}
		\end{equation}
		satisfies the identity
		\begin{equation}
		\mathcal{H}_{2}\left(\gamma\right)=\chi\left(\gamma\right)\mathcal{U}\left(\gamma\right)+\left(1-\chi\left(\gamma\right)\right)\mathcal{L}\left(\gamma\right).
		\end{equation}
		
		That is, it is sufficient to look for a tight approximation for $\left(\ref{XX}\right)$
		so as to approximate $\mathcal{H}_{2}\left(\gamma\right).$
		
		By plotting $\chi\left(\gamma\right)$ as shown in Fig. \ref{compCoeff},
		one can clearly notice its exponential behavior. It follows that its
		approximate expression can be written in the form (\ref{formof X}).
		Furthermore, the optimized coefficients $\mathcal{C}_{i}$ and $\mathcal{D}_{i}$
		outlined in Table \ref{optimumparametr}, for various SNR\ intervals
		can be straight forward obtained using a curve fitting tool (e.g.,
		Matlab Curve Fitting app); this ends the proof. \end{proof}
	
	\begin{remark} It is worthwhile that the first-order MQF can be approximated,
		relying on (\ref{pep}) and (\ref{approx}), by
		\begin{eqnarray}
		\widetilde{Q}_{1}\left(a\sqrt{\gamma},b\sqrt{\gamma}\right) & = & \widetilde{\mathcal{H}}_{0}\left(\gamma\right)+\frac{1}{2}I_{0}\left(\sqrt{2}\gamma\right)\exp\left(-2\gamma\right)\nonumber \\
		& = & \mathcal{K}(a,b,\gamma)\left(\eta\widetilde{\chi}\left(\gamma\right)+\delta\right)\nonumber \\
		&  & +I_{0}\left(\sqrt{2}\gamma\right)\exp\left(-2\gamma\right)\widetilde{\chi}\left(\gamma\right).
		\end{eqnarray}
	\end{remark}
	
	Both exact and approximate functions\ $\chi\left(\gamma\right)$
	and $\widetilde{\chi}\left(\gamma\right)$ are plotted in Fig. \ref{compCoeff}.
	One can observe that there exists a strong matching between the two
	curves over the entire range of $\gamma.$
	
	\begin{figure}[th]
		\hspace{0cm}\includegraphics[scale=0.46]{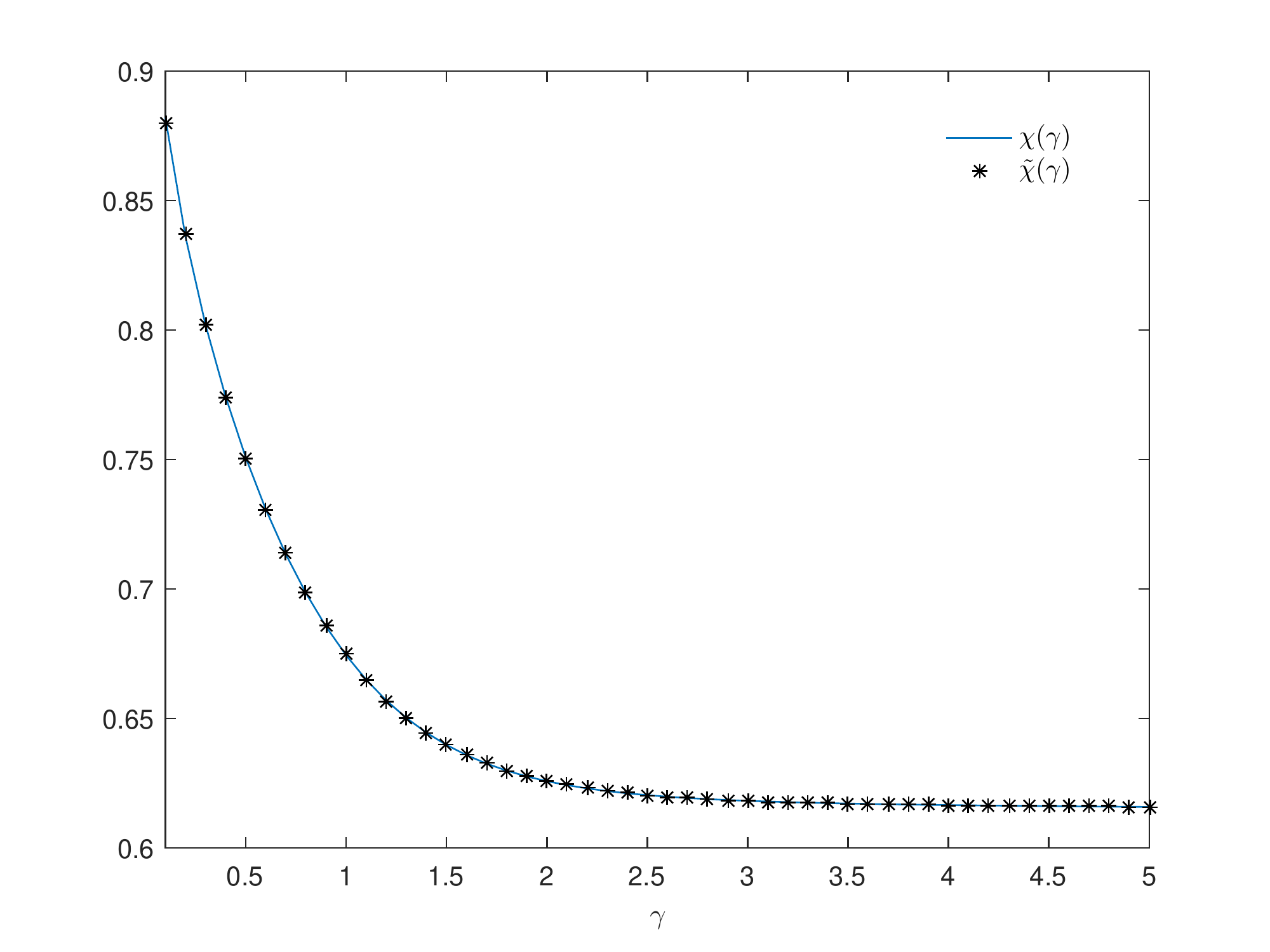} \caption{Comparison between $\chi\left(\gamma\right)$ and $\widetilde{\chi}\left(\gamma\right).$}
		\label{compCoeff}
	\end{figure}
	
	Table \ref{compa} summarizes the accuracy of the proposed approximation
	compared with the best ones proposed in the literature, namely $\left\{ \widetilde{\mathcal{H}}_{i}\left(\gamma\right)\right\} _{i=3..5}$
	labeled $\left\{ BER_{i+2}\right\} _{i=3..5}$ in \cite{barzic},
	respectively. Besides, the relative error corresponding to the aforementioned
	approximations, namely
	
	\[
	{\Large\varepsilon}_{i}=\frac{\left\vert \widetilde{\mathcal{H}}_{i}\left(\gamma\right)-\mathcal{H}_{2}\left(\gamma\right)\right\vert }{\mathcal{H}_{2}\left(\gamma\right)},i=2..5.,
	\]
	is depicted in Fig. \ref{abserror}. Obviously, the relative error
	corresponding to the proposed approximation outperforms those in \cite{barzic},
	except for a short interval, i.e., $[11.7,12.45]$ where $\mathcal{H}_{2}\left(\gamma\right)$
	is negligible, as outlined in Table II, compared to its values getting
	for small SNR.
	
	\begin{table*}
		\caption{Comparison between the exact and approximate {BEP}}
		
		\begin{center}
			\begin{tabular}{c|c|c|c|c|c}
				\hline
				$\gamma$  & Eq. ($\ref{pep})$  & $\widetilde{H}_{2}\left(\gamma\right),$ Eq. ($\ref{approx})$  & $\widetilde{H}_{3}\left(\gamma\right),$ \cite{barzic}  & $\widetilde{H}_{4}\left(\gamma\right),$ \cite{barzic}  & $\widetilde{H}_{5}\left(\gamma\right),$ \cite{barzic} \tabularnewline
				\hline
				$0.5$  & $2.6929\times10^{-1}$  & $2.6918\times10^{-1}$  & $2.7792\times10^{-1}$  & $2.6921\times10^{-1}$  & $2.6911\times10^{-1}$ \tabularnewline
				\hline
				$1$  & $1.6391\times10^{-1}$  & $1.6395\times10^{-1}$  & $1.6383\times10^{-1}$  & $1.6383\times10^{-1}$  & $1.6568\times10^{-1}$ \tabularnewline
				\hline
				$1.5$  & $1.0646\times10^{-1}$  & $1.0645\times10^{-1}$  & $1.0667\times10^{-1}$  & $1.0655\times10^{-1}$  & $1.0667\times10^{-1}$ \tabularnewline
				\hline
				$2$  & $7.1611\times10^{-2}$  & $7.1614\times10^{-2}$  & $7.1685\times10^{-2}$  & $7.1625\times10^{-2}$  & $7.1885\times10^{-2}$ \tabularnewline
				\hline
				$2.5$  & $4.9177\times10^{-2}$  & $4.9178\times10^{-2}$  & $4.9190\times10^{-2}$  & $4.9174\times10^{-2}$  & $4.9481\times10^{-2}$ \tabularnewline
				\hline
				$3$  & $3.4227\times10^{-2}$  & $3.4226\times10^{-2}$  & $3.4228\times10^{-2}$  & $3.4226\times10^{-2}$  & $3.4482\times10^{-2}$ \tabularnewline
				\hline
				$4$  & $1.7013\times10^{-2}$  & $1.7013\times10^{-2}$  & $1.7015\times10^{-2}$  & $1.7018\times10^{-2}$  & $1.7144\times10^{-2}$ \tabularnewline
				\hline
				$5$  & $8.6484\times10^{-3}$  & $8.6485\times10^{-3}$  & $8.6501\times10^{-3}$  & $8.6501\times10^{-3}$  & $8.7059\times10^{-3}$ \tabularnewline
				\hline
				$6$  & $4.4613\times10^{-3}$  & $4.4613\times10^{-3}$  & $4.4624\times10^{-3}$  & $4.4617\times10^{-3}$  & $4.4859\times10^{-3}$ \tabularnewline
				\hline
				$7$  & $2.3256\times10^{-3}$  & $2.3256\times10^{-3}$  & $2.3263\times10^{-3}$  & $2.3257\times10^{-3}$  & $2.3363\times10^{-3}$ \tabularnewline
				\hline
				$8$  & $1.2219\times10^{-3}$  & $1.2219\times10^{-3}$  & $1.2222\times10^{-3}$  & $1.2219\times10^{-3}$  & $1.2266\times10^{-3}$ \tabularnewline
				\hline
				$9$  & $6.4596\times10^{-4}$  & $6.4597\times10^{-4}$  & $6.4613\times10^{-4}$  & $6.4597\times10^{-4}$  & $6.4594\times10^{-4}$ \tabularnewline
				\hline
				$10$  & $3.4318\times10^{-4}$  & $3.4319\times10^{-4}$  & $3.4327\times10^{-4}$  & $3.4319\times10^{-4}$  & $3.4318\times10^{-4}$ \tabularnewline
				\hline
				$11$  & $1.8307\times10^{-4}$  & $1.8307\times10^{-4}$  & $1.8311\times10^{-4}$  & $1.8307\times10^{-4}$  & $1.8307\times10^{-4}$ \tabularnewline
				\hline
				$12$  & $9.7990\times10^{-5}$  & $9.7990\times10^{-5}$  & $9.8011\times10^{-5}$  & $9.7990\times10^{-5}$  & $9.7990\times10^{-5}$ \tabularnewline
				\hline
			\end{tabular}
			\par\end{center}
		\label{compa}
	\end{table*}
	
	\begin{figure}[th]
		\begin{center}
		\includegraphics[scale=0.46]{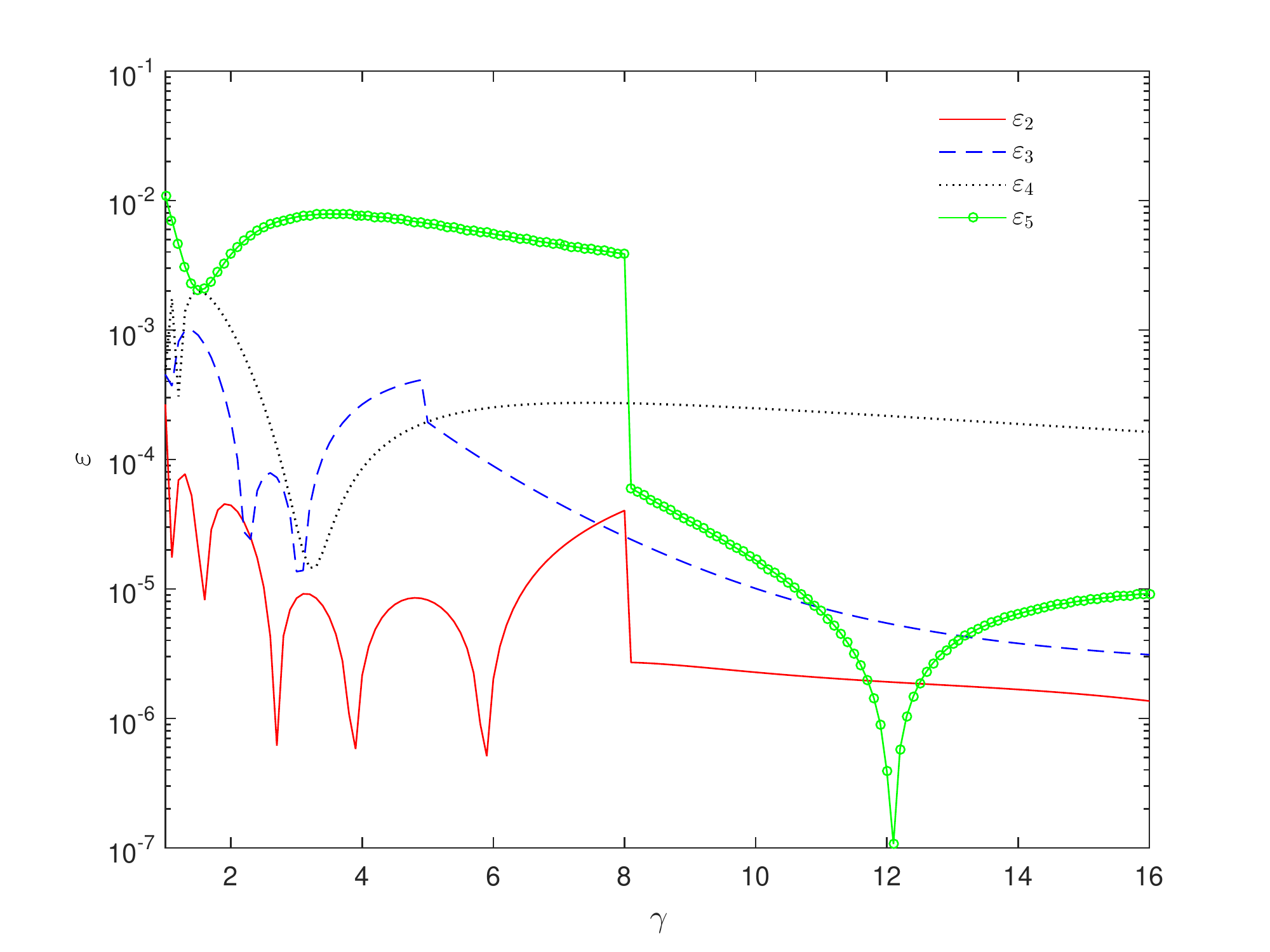} \caption{Comparison of the relative errors.}
				\end{center}
		\label{abserror}
	\end{figure}
	
	\begin{remark} It can be seen clearly that the proposed approximation
		outperforms the concurrent ones. Moreover, one can check that the
		three expressions $\left\{ \widetilde{\mathcal{H}}_{i}\left(\gamma\right)\right\} _{i=3..5}$
		are tough, which rend them useless in numerous applications such as
		the average bit error probability (ABEP) computation under a complicated fading model. Contrarily,
		the expression of $\widetilde{\mathcal{H}}_{2}\left(\gamma\right)$
		is quite simple, making it more appropriate for various fields. \end{remark}
	
	\section{{AEP} Analysis}
	
	As mentioned above, the proposed approximate {EP} is used to derive
	an approximate {AEP} when communicating over $\kappa-%TCIMACRO{\U{b5}}%
	%BeginExpansion
	{\mu}%EndExpansion
	$ shadowed fading.
	
	The PDF of instantaneous SNR $\gamma$ under the $\kappa-%TCIMACRO{\U{b5}}%
	%BeginExpansion
	{\mu}%EndExpansion
	$ shadowed fading model can be written as \cite[Eq. 4]{kapamu}
	\begin{equation}
	f_{\gamma}(\gamma)=\frac{\lambda}{\Gamma\left(\mu\right)}\gamma^{\mu-1}e_{{}}^{-\nu\gamma}{}_{1}F_{1}\left(m;\mu;\omega\gamma\right),\label{pdfKmu}
	\end{equation}
	with
	
	\begin{equation}
	\lambda=\frac{\mu^{\mu}m^{m}\left(1+\mathcal{\kappa}\right)^{\mu}}{\overline{\gamma}^{\mu}\left(\mu\mathcal{\kappa}+m\right)^{m}},\text{ }\nu=\frac{\mu\left(1+\mathcal{\kappa}\right)}{\overline{\gamma}},\text{ }\omega=\frac{\mu^{2}\mathcal{\kappa}\left(1+\mathcal{\kappa}\right)}{\overline{\gamma}\left(\mu\mathcal{\kappa}+m\right)},\label{eq:constants}
	\end{equation}
	where $\overline{\gamma}$ is the average SNR, $\mathcal{\kappa}$
	indicates the power ratio between the dominant waves, $\mu$ refers to
	the scattered components, while $m$ accounts for the shape parameter.
	Further, $_{1}F_{1}\left(.;.;.\right)$ and $\Gamma\left(.\right)$
	denote the Kummer confluent hypergeometric and Euler Gamma functions,
	respectively.
	
	The {AEP} approximation for both $M$-PSK and DQPSK schemes can be
	straightforwardly evaluated as
	\begin{equation}
	{{P_{e}}}=\int_{0}^{\infty}f\left(\gamma\right)\mathcal{H}_{i}\left(\gamma\right)d\gamma,\label{ASERformula}
	\end{equation}
	by setting $i=1$ and $i=2$, respectively.
	\begin{remark}
		{{$P_{e}$} accounts for the average symbol error probability (ASEP) for $M$-PSK modulation, while it refers to the ABEP for DQPSK modulation with Gray coding.}
	\end{remark}
	\begin{proposition} In the case of $M$-PSK modulation, the ASEP
		over $\kappa-%TCIMACRO{\U{b5}}%
		%BeginExpansion
		{\mu}%EndExpansion
		$ shadowed fading model can be approximated as
		\begin{equation}
		{{P_{e}}}\simeq\lambda\sum_{l=1}^{7}\mathcal{A}_{l}\left(\nu+\mathcal{B}_{l}\right)^{-\mu}\left(1-\frac{\omega}{\nu+\mathcal{B}_{l}}\right)^{-m}.\label{ASERMPSK}
		\end{equation}
	\end{proposition}
	
	\begin{proof} By plugging $\left(\ref{APProx-MPSK}\right)$ and $\left(\ref{pdfKmu}\right)$
		in $\left(\ref{ASERformula}\right),$and with the help of \cite[7.522 Eq. (9) and 9.121 Eq. (1)]{table integral}
		one can obtain $\left(\ref{ASERMPSK}\right).$ \end{proof}
	
	\begin{proposition} The {ABEP} for DQPSK scheme with Gray coding over
		the $\kappa-%TCIMACRO{\U{b5}}%
		%BeginExpansion
		{\mu}%EndExpansion
		$ shadowed fading channel can be tightly approximated as
		
		\begin{equation}
		{{P_{e}}}\simeq\frac{\lambda\eta}{\Gamma\left(\mu\right)}\sum_{i=0}^{2}\left[\mathcal{C}_{i}\mathcal{M}_{i,k}^{\left(1\right)}(a)-\mathcal{C}_{i}\mathcal{M}_{i,k}^{\left(1\right)}(-a)+\frac{\mathcal{F}_{i}}{\eta}\mathcal{M}_{i,k}^{\left(2\right)}\right],\label{approx-ABER}
		\end{equation}
		with
		\begin{eqnarray}
		\mathcal{M}_{i,k}^{\left(1\right)}(a) & = & \sum_{k=0}^{\infty}\frac{\phi_{i,k}\text{ }_{2}F_{1}\left(\frac{1}{2},\mu+k;\mu+k+1;\frac{2\xi_{i}}{\left(b-a\right)^{2}+2\xi_{i}}\right)}{2\sqrt{\pi}},\label{integ1}\\
		\mathcal{M}_{i,k}^{\left(2\right)} & = & \sum_{k=0}^{\infty}\psi_{i,k}\text{ }_{2}F_{1}\left(\frac{\mu+k}{2},\frac{\mu+k+1}{2};1;\frac{2}{\left(2+\xi_{i}\right)^{2}}\right),\label{integ2}
		\end{eqnarray}
		
		\begin{eqnarray}
		\phi_{i,k} & = & \frac{\left(m\right)_{k}\omega^{k}}{\left(\mu\right)_{k}k!}\left(\frac{2}{\left(b-a\right)^{2}+2\xi_{i}}\right)^{\mu+k}\frac{\Gamma\left(\mu+k+\frac{1}{2}\right)}{\mu+k},\\
		\psi_{i,k} & = & \frac{\Gamma\left(\mu\right)\left(m\right)_{k}\omega^{k}}{k!\left(2+\xi_{i}\right)^{\mu+k}},
		\end{eqnarray}
		
		\begin{equation}
		\xi_{i}=\nu+\mathcal{D}_{i},
		\end{equation}
		
		\begin{equation}
		\eta=\frac{1-\delta^{2}}{\delta},
		\end{equation}
		where $\left(.\right)_{.}$ represents the pochhammer symbol, $_{2}F_{1}\left(.,.;.;.\right)$
		denotes the hypergeometric functions \cite[Eq. (8.310)]{table integral},
		$\mathcal{F}_{0}=\mathcal{C}_{0},\mathcal{F}_{1}=\mathcal{C}_{1},\mathcal{F}_{2}=-\frac{1}{2},\mathcal{F}_{2}=\frac{\delta^{2}}{1-\delta^{2}},$
		and $\mathcal{D}_{2}=0.$ \end{proposition}
	
	\begin{proof} First, one can check using $\left(\ref{lower}\right)$
		and $\left(\ref{upper}\right)$ jointly with $\left(\ref{KK}\right),$
		$\left(\ref{approx}\right),$ and $\left(\ref{formof X}\right)$
		\begin{equation}
		\widetilde{\mathcal{H}}_{0}\left(\gamma\right)=\eta\sum_{i=0}^{2}\left[\mathcal{C}_{i}e^{-\mathcal{D}_{i}\gamma}\mathcal{K}(a,b,\gamma)+\mathcal{F}_{i}e^{-\left(2+\mathcal{D}_{i}\right)\gamma}I_{0}\left(\sqrt{2}\gamma\right)\right].\label{Hdelta2}
		\end{equation}
		On the other hand, the {ABEP} for DQPSK can be evaluated as
		\begin{equation}
		P_{s}=\int_{0}^{\infty}f\left(\gamma\right)\mathcal{H}_{2}\left(\gamma\right)d\gamma.\label{ABER}
		\end{equation}
		Now, using $\left(\ref{pdfKmu}\right)$ and $\left(\ref{Hdelta2}\right)$
		in $\left(\ref{ABER}\right),$ the {ABEP} can be approximated by
		
		\begin{equation}
		{P_{e}}\simeq\frac{\lambda\eta}{\Gamma\left(\mu\right)}\sum_{i=0}^{2}\left[\mathcal{C}_{i}\mathcal{M}_{i,k}^{\left(1\right)}(a)-\mathcal{C}_{i}\mathcal{M}_{i,k}^{\left(1\right)}(-a)+\frac{\mathcal{F}_{i}}{\eta}\mathcal{M}_{i,k}^{\left(2\right)}\right]
		\end{equation}
		where %%%%%%%%%%%%%%%%%%%
		\begin{eqnarray}
		\mathcal{M}_{i}^{\left(1\right)}(a)=\int_{0}^{\infty}\gamma^{\mu-1}\text{ }_{1}F_{1}\left(m;\mu;\omega\gamma\right)e^{-\left(\mathcal{D}_{i}+v\right)\gamma}\nonumber \\
		\times Q\left(\left(b-a\right)\sqrt{\gamma}\right)d\gamma,\label{integral1}
		\end{eqnarray}
		
		and
		\begin{equation}
		\mathcal{M}_{i}^{\left(2\right)}=\int_{0}^{\infty}\gamma^{\mu-1}\text{ }_{1}F_{1}\left(m;\mu;\omega\gamma\right)e^{-\left(2+\xi_{i}\right)\gamma}I_{0}\left(\sqrt{2}\gamma\right)d\gamma.\label{integ22}
		\end{equation}
		Utilizing Craig's formula of the Gaussian $Q$-function \cite[Eq. (5)]{craigs}
		and \cite[9.14.1]{table integral}, $\left(\ref{integral1}\right)$
		can be written as
		\begin{eqnarray}
		\mathcal{M}_{i}^{\left(1\right)}(a) & = & \sum_{k=0}^{\infty}\frac{\left(m\right)_{k}}{\left(\mu\right)_{k}k!}\frac{\omega^{k}}{\pi}\int_{0}^{\frac{\pi}{2}}\int_{0}^{\infty}\gamma^{\mu+k-1}\nonumber \\
		&  & \times\exp\left(-\left(\frac{\left(b-a\right)^{2}}{2\sin^{2}\left(\theta\right)}+\xi_{i}\right)\gamma\right)d\gamma d\theta\label{integrl11}
		\end{eqnarray}
		By evaluating the inner integral in $\left(\ref{integrl11}\right)$
		with the aid of \cite[Eq. (3.381.4)]{table integral}, we get
		\begin{eqnarray}
		\mathcal{M}_{i}^{\left(1\right)}(a) & = & \sum_{k=0}^{\infty}\frac{\left(m\right)_{k}\omega^{k}\Gamma\left(\mu+k\right)\int_{0}^{\frac{\pi}{2}}\left(\frac{\left(b-a\right)^{2}}{2\sin^{2}\left(\theta\right)}+2\xi_{i}\right)^{-\mu-k}d\theta}{\pi\left(\mu\right)_{k}k!}\nonumber \\
		& = & \sum_{k=0}^{\infty}\frac{\left(m\right)_{k}\omega^{k}\Gamma\left(\mu+k\right)2^{\mu+k}}{\pi\left(\mu\right)_{k}k!}\nonumber \\
		&  & \times\int_{0}^{\frac{\pi}{2}}\left(\frac{\sin^{2}\left(\theta\right)}{\left(b-a\right)^{2}+2\xi_{i}-2\xi_{i}\cos^{2}\left(\theta\right)}\right)^{\mu+k}d\theta\text{ }\label{intgr1_1}
		\end{eqnarray}
		Now, by substituting \cite[Eq. (3.682)]{table integral} into $\left(\ref{intgr1_1}\right),$
		$\left(\ref{integ1}\right)$ is obtained. Finally, $\left(\ref{integ22}\right)$
		can be evaluated relying on \cite[Eqs. (9.14.1) and (6.621.1)]{table integral}
		to obtain $\left(\ref{integ2}\right)$, which concludes the proof.
	\end{proof}
	
	\subsection{Asymptotic Analysis}
	
	In order to gain further insights into system parameters at high SNR
	regime, an asymptotic analysis for the SNR is carried out. Firstly,
	note that for large values of $\overline{\gamma}$, one can see that
	$\omega$ goes to $0$ and thus the term $k=0$ dominates the others,
	$v$ also goes to $0$ (i.e., $\xi_{i}\simeq\mathcal{D}_{i}$). It
	follows that the {AEP} can be asymptotically approximated as
	
	\begin{equation}
	{P_{e}}\sim\lambda\sum_{l=1}^{7}\mathcal{A}_{l}\mathcal{B}_{l}^{-\mu}.\label{asyptMPSK}
	\end{equation}
	and %%%%%%%%%%%%%
	\begin{equation}
	{P_{e}}\sim\frac{\lambda\eta}{\Gamma\left(\mu\right)}\sum_{i=0}^{2}\left[\begin{array}{c}
	\mathcal{C}_{i}\mathcal{M}_{i}^{\left(1,asy\right)}(a)-\mathcal{C}_{i}\mathcal{M}_{i}^{\left(1,asy\right)}(-a)\\
	+\frac{\mathcal{F}_{i}}{\eta}\mathcal{M}_{i}^{\left(2,asy\right)}
	\end{array}\right],\label{asymptotic}
	\end{equation}
	
	for $M$-PSK and DQPSK schemes, respectively, with
	\begin{eqnarray}
	\mathcal{M}_{i}^{\left(1,asy\right)}(a) & \sim & \Delta\text{ }_{2}F_{1}\left(\frac{1}{2},\mu;\mu+1;\frac{2\mathcal{D}_{i}}{\left(b-a\right)^{2}+2\mathcal{D}_{i}}\right),\label{asymp1}\\
	\Delta & = & \frac{1}{2\sqrt{\pi}}\left(\frac{2}{\left(b-a\right)^{2}+2\mathcal{D}_{i}}\right)^{\mu}\frac{\Gamma\left(\mu+\frac{1}{2}\right)}{\mu},\\
	\mathcal{M}_{i}^{\left(2,asy\right)} & \sim & \frac{\Gamma\left(\mu\right)}{\left(2+\mathcal{D}_{i}\right)^{\mu}}\text{ }_{2}F_{1}\left(\frac{\mu}{2},\frac{\mu+1}{2};1;\frac{2}{\left(2+\mathcal{D}_{i}\right)^{2}}\right).\label{asymp2}
	\end{eqnarray}
	
	It is worth mentioning from \eqref{asyptMPSK} and \eqref{asymptotic}
	alongside with \eqref{eq:constants} that the diversity order equals
	$\mu$.
	
	\subsection{Bound on the truncation error}
	
	The above approximate {ABEP} for DQPSK is expressed in terms of infinite
	series. Truncating such summation and estimating the truncated error
	is though of paramount importance for numerical evaluation purposes.
	In what follows, a closed-form bound for such truncation error is provided.
	
	Using $\left(\ref{integ1}\right),$ the truncation up to $L-1$ terms
	of the first summation results to the following error
	\begin{equation}
	\epsilon_{i}^{(1)}(a)=\sum_{k=L}^{\infty}\frac{\phi_{k}}{2\sqrt{\pi}}\text{ }_{2}F_{1}\left(\frac{1}{2},\mu+k;\mu+k+1;\frac{2\xi_{i}}{\left(b-a\right)^{2}+2\xi_{i}}\right)\label{erorinteg1}
	\end{equation}
	
	By changing the summation index to $j=k-L$ in $\left(\ref{erorinteg1}\right)$,
	then using \cite[Eq. (06.10.02.0001.01)]{wolframe}, and performing
	some manipulations, the bound can be expressed as
	\begin{eqnarray}
	\epsilon_{i}^{(1)}(a) & \leq & \frac{\Theta_{i,L}(a)}{2\sqrt{\pi}}\text{ }_{1}F_{0}\left(\frac{1}{2};-;\frac{2\xi_{i}}{\left(b-a\right)^{2}+2\xi_{i}}\right)\label{erorinteg1Bound}\\
	&  & \times\text{ }_{2}F_{1}\left(1,m+L;L+1;\frac{2\omega}{\left(b-a\right)^{2}+2\xi_{i}}\right),\nonumber
	\end{eqnarray}
	
	with
	\begin{eqnarray}
	\Theta_{i,L}(a) & = & \left(\frac{2}{\left(b-a\right)^{2}+2\xi_{i}}\right)^{\mu}\frac{\Gamma\left(\mu\right)\Gamma\left(m+L\right)}{L!\Gamma\left(m\right)}\\
	&  & \times\left(\frac{2\omega}{\left(b-a\right)^{2}+2\xi_{i}}\right)_{{}}^{L}.\nonumber
	\end{eqnarray}
	
	In a similar manner, the truncated error of the summation $\left(\ref{integ2}\right)$
	can be upper bounded by
	
	\begin{eqnarray}
	\epsilon_{i}^{(2)} & \leq & \Lambda_{i,L}\text{ }_{2}F_{1}\left(\frac{\mu+L}{2},\frac{\mu+L+1}{2};1;\frac{2}{\left(2+\xi_{i}\right)^{2}}\right)\label{erorinteg2bound}\\
	&  & \times\text{ }_{2}F_{1}\left(1,m+L;L+1;\frac{\omega}{2+\xi_{i}}\text{ }\right),\nonumber
	\end{eqnarray}
	
	with
	\begin{equation}
	\Lambda_{i,L}=\frac{\Gamma\left(\mu\right)\omega^{L}}{L!\Gamma\left(m+L\right)\Gamma\left(m\right)\left(2+\xi_{i}\right)^{\mu+L}}.
	\end{equation}
	
	Consequently, and having in mind that $\epsilon_{i}^{(j)}$ are positives,
	as can be seen from \ref{erorinteg1Bound} and \ref{erorinteg2bound},
	the absolute value of the total truncated error can be upper bounded
	by
	\begin{equation}
	\left|\epsilon_{P_{b}}\right|\leq\frac{\lambda\eta}{\Gamma\left(\mu\right)}\sum_{i=0}^{2}\left[\mathcal{C}_{i}\epsilon_{i}^{(1)}(a)+\mathcal{C}_{i}\epsilon_{i}^{(1)}(-a)+\frac{\mathcal{F}_{i}}{\eta}\epsilon_{i}^{(2)}\right].\label{erorASER}
	\end{equation}

	\section{Results and Discussion}
	
	In this section, the proposed approximation for the {AEP} versus SNR
	(in dB) for both $M$-PSK and DQPSK modulation schemes over $\kappa-%TCIMACRO{\U{b5}}%
	%BeginExpansion
	{\mu}%EndExpansion
	$ shadowed fading channel is evaluated and compared with the exact one
	for various fading severity parameters.
	\begin{itemize}
		
		\item Figs. \ref{weakmpsk} and \ref{weakdqpsk} illustrate, for a fixed
		value of $m$, the effect of parameter $\mu$ on the {AEP} for $M$-PSK
		and DQPSK modulation techniques, respectively under a weak line of sight (LOS)
		condition. One can notice that the greater the $\mu$ is, the better
		the system's performance.
		
		\item Figs. \ref{strongmpsk} and \ref{strongdqpsk} depict the {AEP} for
		both considered modulation schemes under strong LOS ($\kappa=10$) for a fixed
		value of $\mu$. It is observed that countering the effect of shadowing
		requires the increase of $m$.
		
		\item To show the versatility of the $\kappa-\mu$ shadowed fading, Figs.
		\ref{nakampsk}, \ref{classmpsk}, and \ref{classdqpsk} present the
		{AEP} for some classical fading models. Noteworthy, the results in
		all figures are provided for either integer or non-integer values
		of $\mu$ and $m$. Further, the simulation curves match perfectly
		with the proposed approximation.
		
		\item Fig. \ref{Truncerror} depicts the absolute value of the truncated
		error versus the number of limited terms $L$ for DQPSK modulation.
		It can be shown that the greater $L$ is, the smaller such an error. Interestingly,
		the truncated error decreases with the increase of SNR.
		
		\item Lastly, Fig. \ref{diversity} presents the achievable diversity order for both
		modulation schemes versus the average SNR, computed by evaluating
		$-\frac{\log P_{s}}{\log\bar{\gamma}}$. It is clearly noticed that
		such a metric goes to $\mu$ as $\bar{\gamma}$ tends to infinity.
		
	\end{itemize}
	\begin{figure}[th]
		\begin{center}
			\hspace{0cm}\includegraphics[scale=0.46]{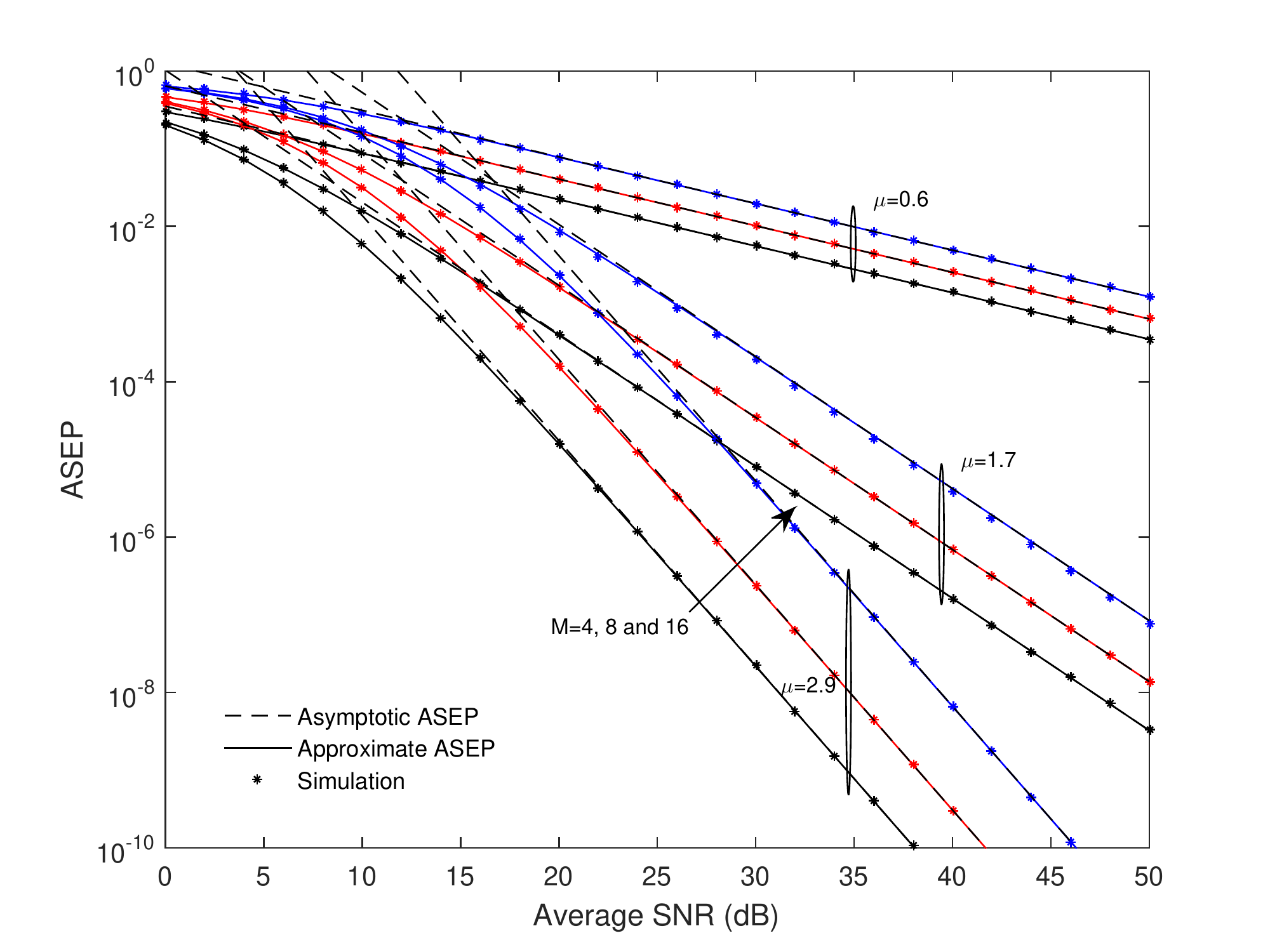} \
		\end{center}
		\caption{ASEP for $M$-PSK under weak LOS scenario ($\kappa=1)$ with different
			values of $\mu$ and $m=1.3.$}
		\label{weakmpsk}
	\end{figure}
	
	\begin{figure}[th]
		\begin{center}
			\hspace{0cm}\includegraphics[scale=0.46]{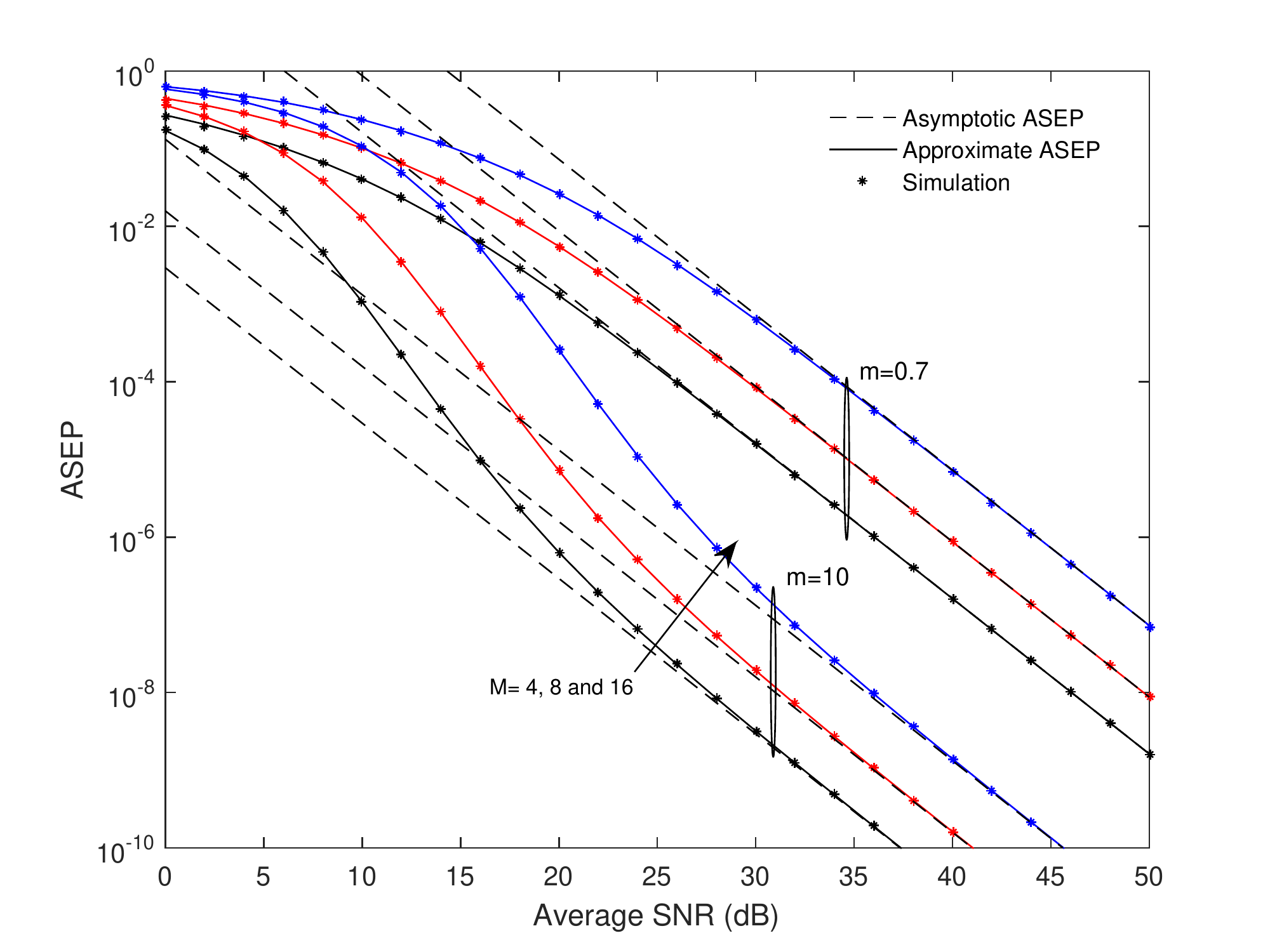} \
		\end{center}
		\caption{ASEP for $M$-PSK under strong LOS scenario ($\kappa=10)$ with different
			values of $m$ and $\mu=2.$}
		\label{strongmpsk}
	\end{figure}
	
	\begin{figure}[th]
		\begin{center}
			\hspace{0cm}\includegraphics[scale=0.46]{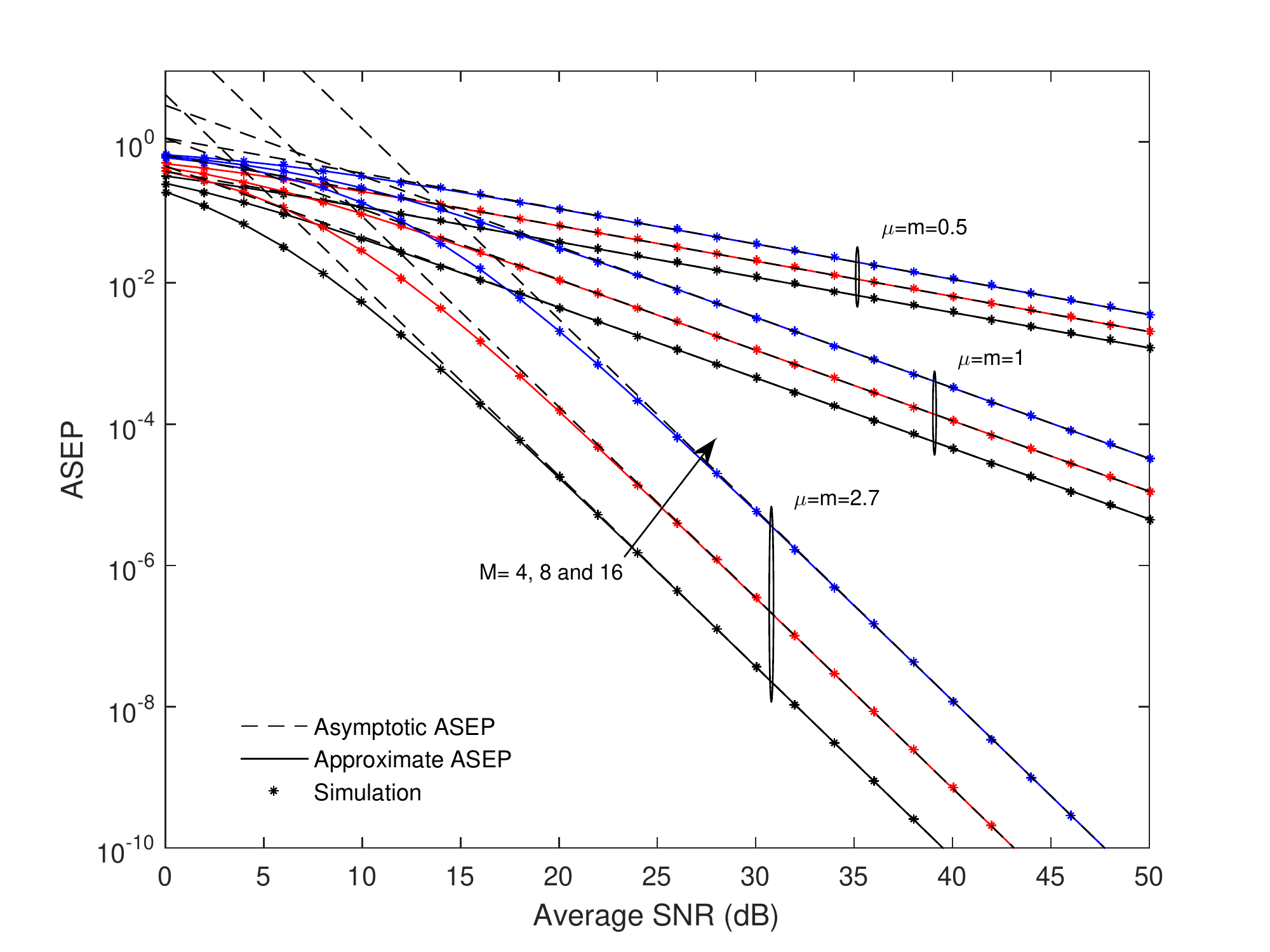} \
		\end{center}
		\caption{ASEP for $M$-PSK under non-LOS scenario ($\kappa=0)$ with $\mu=m.$}
		\label{nakampsk}
	\end{figure}
	
	\begin{figure}[th]
		\begin{center}
			\hspace{0cm}\includegraphics[scale=0.46]{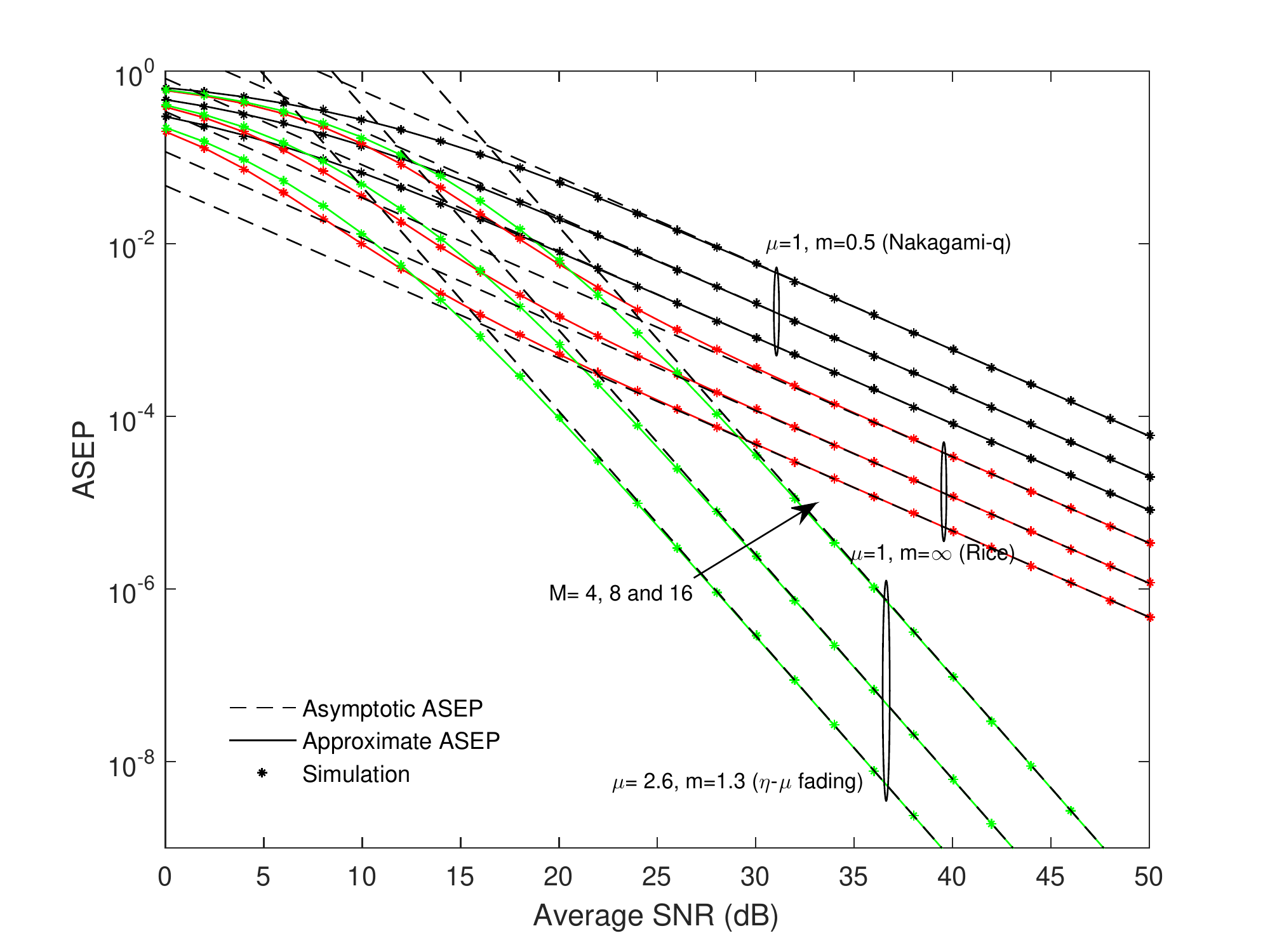} \
		\end{center}
		\caption{ASEP for $M$-PSK over various practical fading models with $\kappa=5.$}
		\label{classmpsk}
	\end{figure}
	
	\begin{figure}[th]
		\begin{center}
			\hspace{0cm}\includegraphics[scale=0.46]{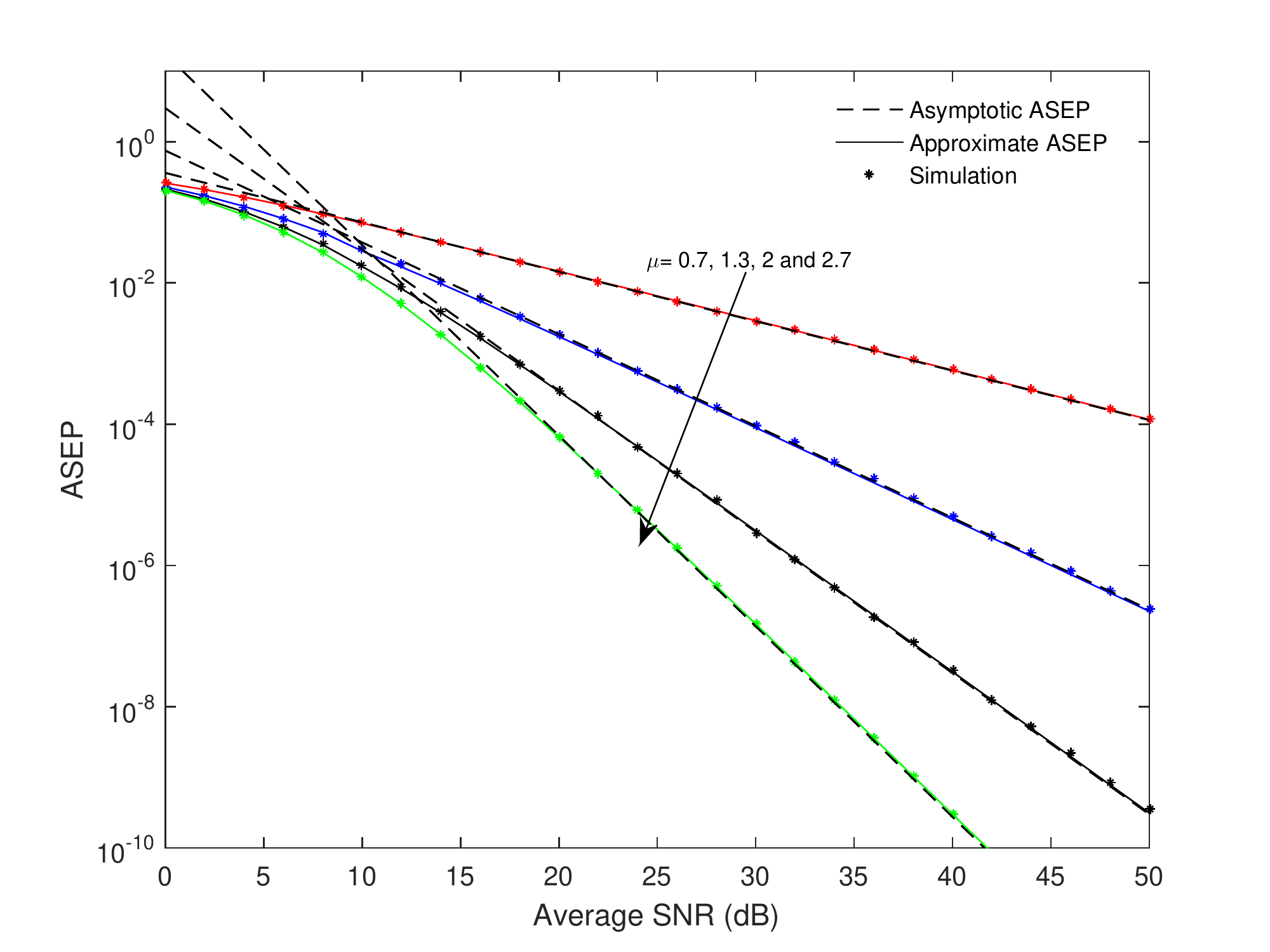} \
		\end{center}
		\caption{ABEP for DQPSK under weak LOS scenario ($\kappa=1)$ with different
			values of $\mu$ and $m=1.3.$}
		\label{weakdqpsk}
	\end{figure}
	
	\begin{figure}[th]
		\begin{center}
			\hspace{0cm}\includegraphics[scale=0.46]{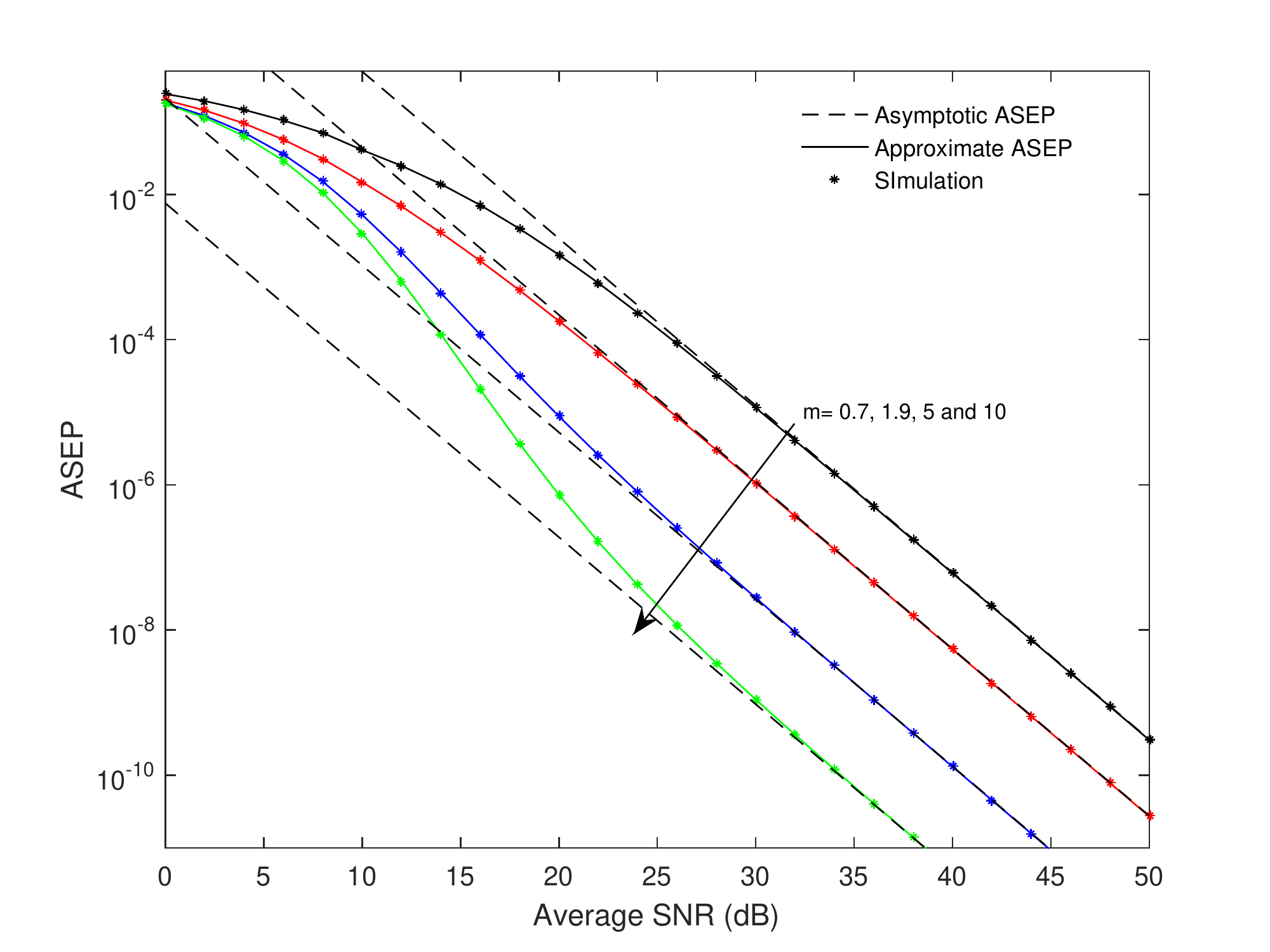} \
		\end{center}
		\caption{ABEP for DQPSK under strong LOS scenario ($\kappa=10)$ with different
			values of $m$ and $%TCIMACRO{\U{b5}}%
			%BeginExpansion
			{\mu}%EndExpansion
			=2.3.$}
		\label{strongdqpsk}
	\end{figure}
	
	\begin{figure}[th]
		\begin{center}
			\hspace{0cm}\includegraphics[scale=0.46]{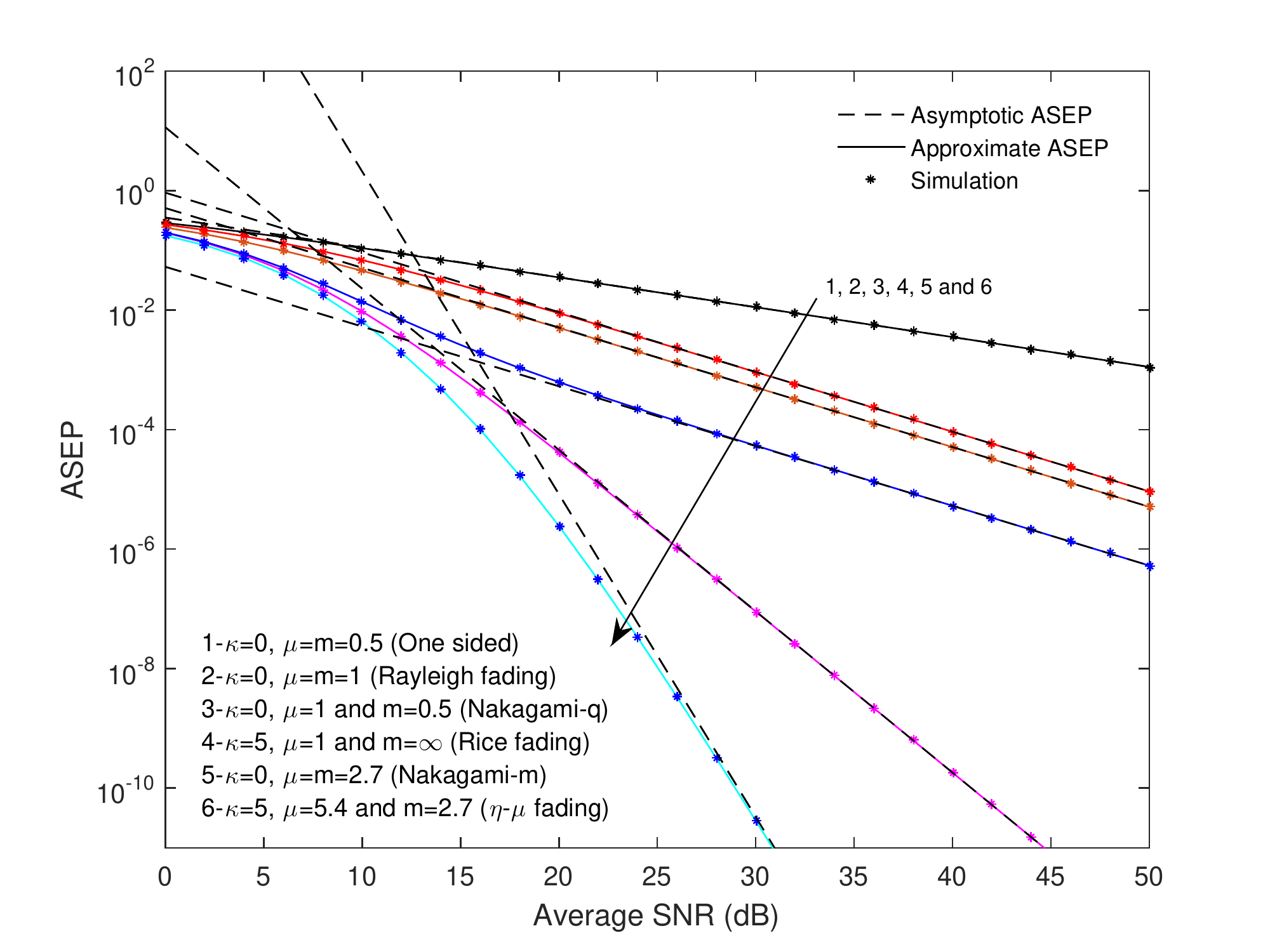} \
		\end{center}
		\caption{ABEP for DQPSK over numerous practical fading distributions with $\kappa=5.$}
		\label{classdqpsk}
	\end{figure}
	
	\begin{figure}[th]
		\begin{center}
			\hspace{0cm}\includegraphics[scale=0.46]{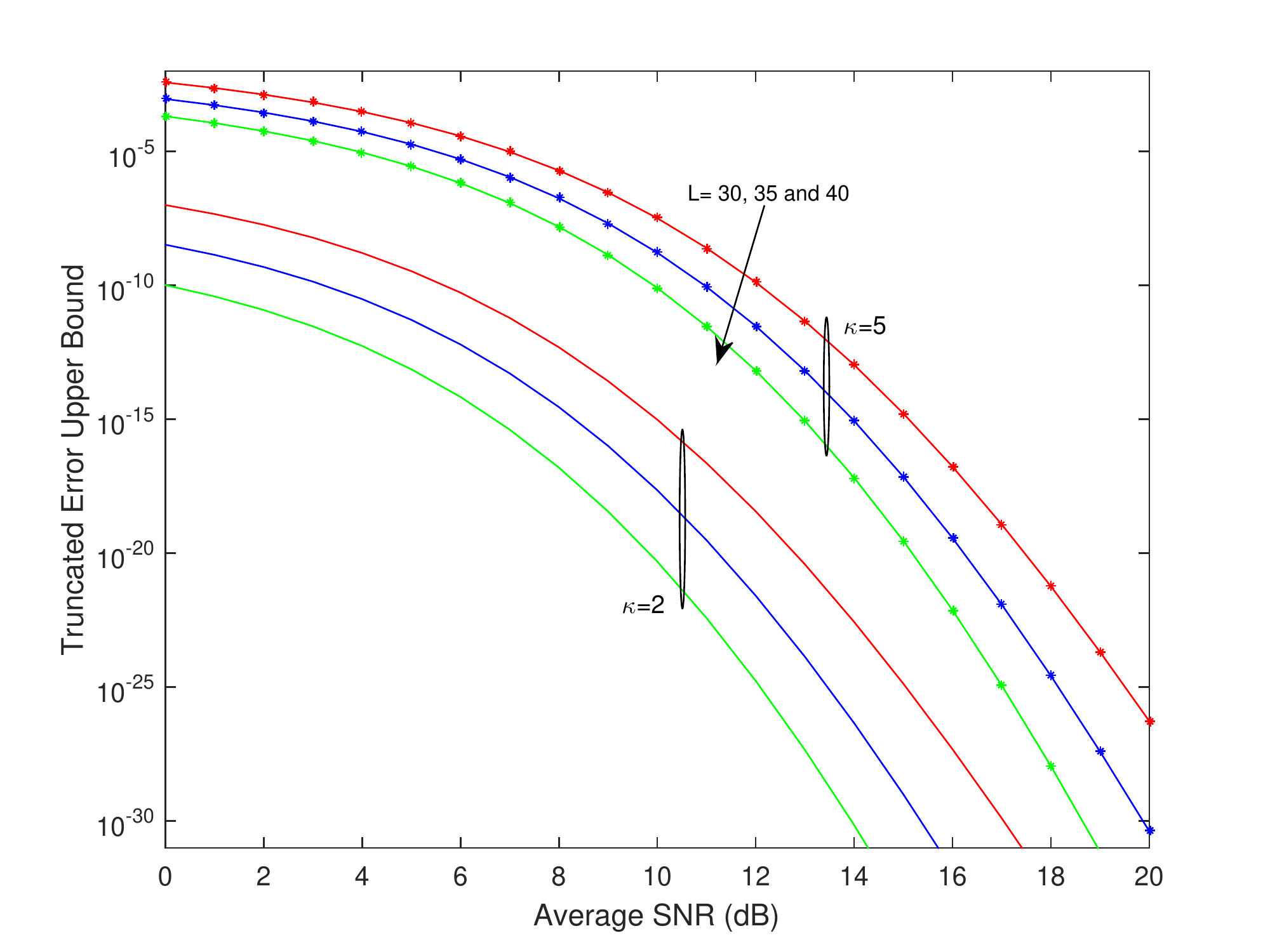} \
		\end{center}
		\caption{Upper bound for the truncated error for $\mu=2.3$ and $m=4.7$ and
			various values of $\kappa$ and $L$.}
		\label{Truncerror}
	\end{figure}
	
	\begin{figure}[th]
		\begin{center}
			\hspace{0cm}\includegraphics[scale=0.46]{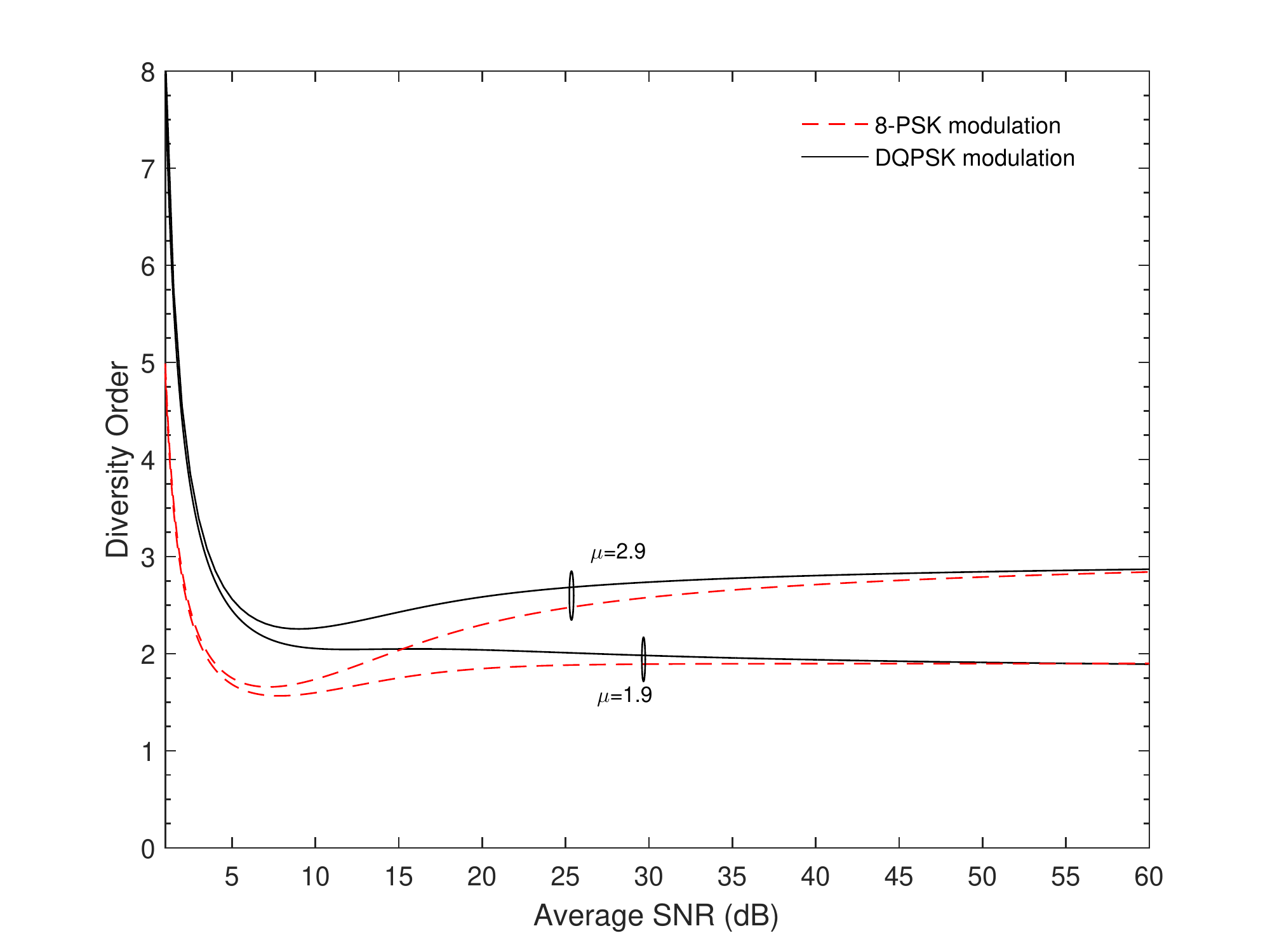} \
		\end{center}
		\caption{Diversity order for $\kappa=5$, $m=4.7$ and various values of $\mu$
			under $M$-PSK and DQPSK modulation schemes.}
		\label{diversity}
	\end{figure}

	\section{Conclusion}
	
	New approximate expressions for the {EP} of a communication system
	employing either $M$-PSK or DQPSK modulation have been derived. The
	proposed approximations ensures optimal accuracy-analytical tractability
	trade-off that enables its versatility to contribute to the {AEP} computation
	over generalized fading channel. The resulting accuracy is better
	than that reached by other existing works relying on more complex
	mathematical expressions. Furthermore, a new closed-form approximation
	for the {AEP} under $\kappa-\mu$ shadowed fading model has been investigated
	and it is accurate for all the practical values of the SNRs, and are
	valid for the entire range of the shaping parameters $\kappa,\mu$,
	and $m$. As far as we know, no previous works dealt with such fading
	and modulation scheme with such a simple approximation. As a future
	aspect, the authors aim to extend the same approach on more general
	fading model such as $\alpha-\kappa-\mu$ shadowed \cite{alphkmushadowed}
	and fluctuating Beckmann fading models \cite{FB}.
	
	\section*{Abbreviations}
	AEP: average error probability; ABEP: average bit error probability; ASEP: average symbol error probability; AWGN: additive white Gaussian noise; DQPSK: differential quaternary phase-shift keying; EP:error probability; LOS: line
	of sight; MQF: Marcum $Q$-function of the first order; MBF: modified
	Bessel function; PDF: probability density function; PSK: phase-shift
	keying modulation schemes; SEP: symbol error probability; SNR: signal-to-noise
	ratio; UIFH: upper incomplete upper Fox's H-function.
	
	\section*{Declarations}
	
	\subsection*{Ethics Approval and Consent to Participate}
	
	The authors declare that this subsection is not applied for this work.
	
	\subsection*{Consent for Publication}
	
	The authors declare that they wrote completely all scripts associated
	with the results presented in this work (i.e., figures and tables). No script
	or data has been imported or used from subsection is not applied for
	this work.
	
	\subsection*{Availability of Data and Material}
	
	All scripts related to this work, developed by the two authors, can be found in github.com/FaissalElBouanani/MarcumQfunctionKappaMu/.
	
	\subsection*{Competing interests}
	
	The authors declare that they have no competing interests.
	
	\subsection*{Funding}
	
	The authors received no specific funding for this work.
	
	\subsection*{Authors\textquoteright{} contributions}
	
	YM derived new approximate expressions for (i) Marcum $Q$-function
	of the first order and (ii) SEP integral-form for $M$-PSK modulation.
	Based on these two results, YM and FE derived the approximate expressions
	for ASEP under both modulation schemes along with their asymptotic
	form and the achievable diversity order. YM performed the simulations.
	YM and FE wrote the paper, analyze and revise the results. All authors
	read and approved the final manuscript.
	
	\subsection*{Acknowledgment}
	
	The paper has been developed and written exclusively by the two authors.
	
	\section*{Figure Legends}
	
	\begin{itemize}
		\item {Fig. 1 Comparison between $\chi\left(\gamma\right)$ and $\widetilde{\chi}\left(\gamma\right)$. To show the accuracy of the fitting method, Fig. 1 depicts the curves of both exact and approximated fitting coefficients.}
		\item {Fig. 2 Comparison of the relative errors. To demonstrate the tightness of the proposed EP's approximate expressions for GC-DQPSK modulation, Fig. 2 depicts the absolute relative error of the proposed approximation with solid line as well as the best ones proposed in the literature.}
		\item {Fig. 3 ASEP for $M$-PSK under weak LOS scenario ($\kappa=1)$ with different values of $\mu$ and $m=1.3$.  The approximated ASEP is presented with solid line, the simulated one with marker and dashed line presents the asymptotic ASEP.}
		\item {Fig. 4 ASEP for M-PSK under strong LOS scenario ($\kappa=10)$ with different values of $m$ and $\mu=2$. The approximated ASEP is presented with solid line, the simulated one with marker and dashed line presents the asymptotic ASEP.}
		\item {Fig. 5 ASEP for M-PSK under non-LOS scenario ($\kappa=0)$ with $\mu=m$. The approximated ASEP is presented with solid line, the simulated one with marker and dashed line presents the asymptotic ASEP.}
		
		\item {Fig. 6 ASEP for M-PSK over various practical fading models with ($\kappa=5)$. The approximated ASEP is presented with solid line, the simulated one with marker and dashed line presents the asymptotic ASEP.}
		
		\item{ Fig. 7 ABEP for DQPSK under weak LOS scenario ($\kappa=1)$ with different values of $\mu$ and $m = 1.3$. The approximated ABEP is presented with solid line, the simulated one with marker and dashed line presents the asymptotic ABEP.}
		
		\item {Fig. 8 ABEP for DQPSK under strong LOS scenario ($\kappa=10)$ with different values of $m$ and $\mu= 2.3$. The approximated ABEP is presented with solid line, the simulated one with marker and dashed line presents the asymptotic ABEP.}
		\item {Fig. 9 ABEP for DQPSK over numerous practical fading distributions with ($\kappa=5)$. The approximated ABEP is presented with solid line, the simulated one with marker and dashed line presents the asymptotic ABEP.}
		\item{Fig. 10 Upper bound for the truncated error for $\mu= 2.3$ and $m = 4.7$ and various values of ($\kappa)$ and $L$.}
		\item{Fig. 11 Diversity order for $\kappa= 5$, $m = 4.7$ and various values of $\mu$ under M-PSK and DQPSK modulation schemes. The diversity curve of $M$-PSK modulation is presented by dashed line while the solid line presents that for GC-DQPSK modulation.}
	\end{itemize}

\end{document}